%% file: 2020-trust-bias.tex
\documentclass[sigconf,natbib,screen=true]{acmart}
\input{packages}
\input{copyright}

\input{definitions}

\input{authors}

\acmSubmissionID{1618}

\title[When Inverse Propensity Scoring does not Work: Affine Corrections for Unbiased Learning to Rank]{When Inverse Propensity Scoring does not Work:\\ Affine Corrections for Unbiased Learning to Rank}

\allowdisplaybreaks

\begin{document}
\fancyhead{}
\begin{abstract}
Besides position bias, which has been well-studied, trust bias is another type of bias prevalent in user interactions with rankings: users are more likely to click incorrectly w.r.t.\ their preferences on highly ranked items because they trust the ranking system. 
While previous work has observed this behavior in users, we prove that existing \ac{CLTR} methods do not remove this bias, including methods specifically designed to mitigate this type of bias.
Moreover, we prove that \ac{IPS} is principally unable to correct for trust bias under non-trivial circumstances.
Our main contribution is a new estimator based on affine corrections: it both reweights clicks and penalizes items displayed on ranks with high trust bias.
Our estimator is the first estimator that is proven to remove the effect of both trust bias and position bias. 
Furthermore, we show that our estimator is a generalization of the existing \ac{CLTR} framework: if no trust bias is present, it reduces to the original \ac{IPS} estimator.
Our semi-synthetic experiments indicate that by removing the effect of trust bias in addition to position bias, \ac{CLTR} can approximate the optimal ranking system even closer than previously possible.
\end{abstract}

% CCS concepts

%\ccsdesc[500]{Information systems~Learning to rank}


\maketitle

\acresetall

\input{sections/01-introduction}
\input{sections/02-Background}

\input{sections/03-Trust}

\input{sections/04-ExistingMethod}

\input{sections/05-Method}
\input{sections/06-Experiments}
\input{sections/07-Results}
\input{sections/08-Conclusion}

\vspace*{-2mm}
\section*{Code and data}
To facilitate the reproducibility of the reported results, this work only made use of publicly available data and our experimental implementation is publicly available at \url{https://github.com/AliVard/trust-bias-CIKM2020}.

\vspace*{-2mm}
\begin{acks}
This research was supported by Elsevier and the Netherlands Organisation for Scientific Research (NWO)
under pro\-ject nrs
652.\-002.\-001 and
612.\-001.\-551.
All content represents the opinion of the authors, which is not necessarily shared or endorsed by their respective employers and/or sponsors.
\end{acks}

\vspace*{-2mm}
%\clearpage
\bibliographystyle{ACM-Reference-Format}
\bibliography{references}

\end{document}

%% file: packages.tex
\PassOptionsToPackage{dvipsnames}{xcolor}

\usepackage{dsfont}
\usepackage{acronym}
\usepackage[noend]{algorithmic}
\usepackage{algorithm}
\usepackage{bbm}
\usepackage{beramono}
\usepackage{bm}
\usepackage{booktabs}
\usepackage[skip=0pt]{caption}
\usepackage{cleveref}
\usepackage[T1]{fontenc}
\usepackage{graphicx}
\usepackage{hyperref}
\usepackage{listings}
\usepackage{multirow}
\usepackage{natbib}
\usepackage{subcaption}
\usepackage{tabularx}
\usepackage[dvipsnames]{xcolor}
\usepackage{enumerate}
\usepackage[inline]{enumitem}
\usepackage{numprint}
\usepackage[normalem]{ulem}
\usepackage{amsthm}
\usepackage{tikz}
\usetikzlibrary{automata, positioning, arrows}
\usepackage{xspace}

\usepackage{mathrsfs}
\usepackage{mathtools}
\usepackage[symbol]{footmisc}

%% file: copyright.tex
\copyrightyear{2020}
\acmYear{2020}
\setcopyright{acmlicensed}\acmConference[CIKM '20]{Proceedings of the 29th ACM International Conference on Information and Knowledge Management}{October 19--23, 2020}{Virtual Event, Ireland}
\acmBooktitle{Proceedings of the 29th ACM International Conference on Information and Knowledge Management (CIKM '20), October 19--23, 2020, Virtual Event, Ireland}
\acmPrice{15.00}
\acmDOI{10.1145/3340531.3412031}
\acmISBN{978-1-4503-6859-9/20/10}

%% file: definitions.tex
\definecolor{ali}{RGB}{0, 150, 0}
\definecolor{mdr}{RGB}{200, 0, 0}
\definecolor{ho}{RGB}{0, 50, 150}

\newcommand{\OurMethod}{affine\xspace}
\newcommand{\OurMethodShort}{affine\xspace}

\acrodef{IR}{Information Retrieval}
\acrodef{LTR}{Learning to Rank}
\acrodef{CLTR}{Counterfactual Learning to Rank}
\acrodef{OLTR}{Online Learning to Rank}
\acrodef{IPS}{Inverse Propensity Scoring}
\acrodef{PBM}{Position-Based Model}
\acrodef{CTR}{Click Trough Rate}
\acrodef{ARP}{Average Relevance Position}
\acrodef{DCG}{Discounted Cumulative Gain}
\acrodef{EM}{Expectation Maximization}
\acrodef{MNAR}{Missing-Not-At-Random}
\acrodef{MAPE}{Mean Absolute Percentage Error}

\newcommand\mydots{\makebox[0.6em][c]{.\hfil.\hfil.}}

\theoremstyle{acmdefinition}

\theoremstyle{acmplain}
\newtheorem{theorem}{Theorem}

\allowdisplaybreaks

\newcommand{\BayesIPS}{\mathit{Bayes\text{-}IPS}}

%% file: authors.tex
\author{Ali Vardasbi$^*$}
\thanks{* Equal contribution.}
\affiliation{%
	\institution{University of Amsterdam}
	\city{Amsterdam}
	\country{The Netherlands}
}
\email{a.vardasbi@uva.nl}

\author{Harrie Oosterhuis$^*$}
\affiliation{%
	\institution{University of Amsterdam}
	\city{Amsterdam}
	\country{The Netherlands}
}
\email{oosterhuis@uva.nl}
 
\author{Maarten de Rijke}
\orcid{0000-0002-1086-0202}
\affiliation{
 \institution{\mbox{}\hspace*{-5mm}\mbox{University of Amsterdam \& Ahold Delhaize}}
 \city{Amsterdam}
 \country{The Netherlands}
}
\email{derijke@uva.nl}

\def\authors{Ali Vardasbi, Harrie Oosterhuis, and Maarten de Rijke}

%% file: sections/01-introduction.tex
% !TEX root = ../2020-trust-bias.tex

\section{Introduction}
\label{sec:intro} 
% supervised LTR
\ac{LTR} is a long-established area of research that continues to receive considerable attention from academia and industry~\citep{liu2009learning}.
Supervised approaches to \ac{LTR} use manually annotated data, where human annotators have provided relevance judgements.
Over time, the limitations of such approaches have become apparent: manually annotated labels are time consuming and expensive to create~\citep{qin2013introducing, Chapelle2011}; moreover, the preferences of actual users and annotators need not be aligned~\citep{sanderson2010}.
Instead, recent years have brought increased interest in \ac{LTR} methods that learn from user interactions.

% user interactions, noise and bias
At first glance user interactions have great advantages over labelled data: online search engines receive large numbers of interactions at virtually no additional costs; and interactions reflect actual user preferences as opposed to annotators' preferences.
Unfortunately, user interactions also bring their own difficulties because they are a form of noisy and biased implicit feedback.
For instance, clicks are noisy in the sense that, often, a non-relevant item receives a click or a relevant item is skipped.
The effect of noise is easily mitigated by averaging over a large number of clicks, but this is not true for bias.
\emph{Position bias}, a well-known type of bias of interactions through clicks~\citep{craswell2008experimental}, occurs because users are more likely to examine results at higher ranks.
As a consequence, an item may receive more clicks because it was displayed at a high rank, not because it was preferred by the user.
Other types of bias include \emph{item-selection bias}: not all items can be displayed at once~\citep{ovaisi2020correcting, oosterhuis2020topkrankings}; \emph{presentation bias}: items are presented in different manners~\citep{yue-2010-beyond}; and \emph{trust-bias}: users are more likely to click incorrectly on higher ranked items~\citep{joachims2005accurately}.
In order to infer a user's true preferences from their interactions, the effects of these biases have to be corrected for.

% OLTR and CLTR
% Harrie: I think OLTR does not have to be discussed in the introduction, since it is not directly relevant to our work, we can move this to the related work.
%There are two general classes of \ac{LTR} from user interactions.
%In the \ac{OLTR}, the optimal ranker is learned through adaptive interventions in the search results \todo{cite}.
%Theses interventions may cause some sorts of user dissatisfaction.
%In order to prevent any risks of reduced quality in the online sessions, \ac{CLTR} algorithms are proposed which use the historical interactions as their supervision source.
%This paper belongs to the \ac{CLTR} class of algorithms.

% IPS
Research into \ac{CLTR} aims to find methods that learn from user interactions but whose optimization process is unaffected by biases~\citep{joachims2017unbiased}.
Early \ac{CLTR} methods correct for position bias using \ac{IPS}~\citep{wang2016learning, joachims2017unbiased}.
\ac{IPS} estimators weight clicks inversely to the probability of the clicked items being examined during logging.
Thus, clicks on items that are less likely to have been examined by users are weighted more heavily.
This reweighting compensates for the effect of position bias, allowing \ac{CLTR} methods to estimate and learn without being affected by position bias in expectation.
Later \ac{CLTR} work has focused on estimating examination probabilities~\citep{ai2018unbiased, agarwal2019estimating, wang2018position, fang2019intervention}, training deep learning models~\citep{agarwal2019counterfactual}, and correcting for more types of bias~\citep{ovaisi2020correcting, oosterhuis2020topkrankings, agarwal2019addressing}.
In particular, \citet{agarwal2019addressing} have proposed an expansion to \ac{IPS} to correct for both position bias and trust bias.

% trust bias problem
In this paper, we prove that \emph{no \ac{IPS} estimator is able to correct for trust bias}, under non-trivial circumstances.
Since all existing bias mitigation methods are \ac{IPS}-based approaches, this implies that there is currently no known \ac{CLTR} method that can deal with trust bias.
We identify the root cause to be the fact that \ac{IPS} only corrects for \ac{MNAR} feedback~\citep{joachims2017unbiased}.
While position bias prevents clicks from occurring due to a lack of user examination, trust bias adds additional clicks due to user trust~\citep{agarwal2019addressing, joachims2005accurately}.
Hence, clicks that are affected by trust bias are not simply a form of \ac{MNAR} feedback and \ac{IPS} cannot correct for such biases.

% trust bias solution
We introduce a novel estimator for \ac{CLTR} that makes use of \emph{affine} corrections, as opposed to the \emph{linear} corrections of \ac{IPS}.
Our novel affine estimator both reweights clicks based on examination probabilities and penalizes items for being displayed on ranks where many incorrect clicks take place.
We prove that the affine estimator is the first method that can correct for both position bias and trust bias.
Furthermore, we show that it is an extension of the existing \ac{CLTR} framework: when no trust bias is present the affine estimator naturally reduces to an \ac{IPS} estimator.
The results of our semi-synthetic experiments show that while existing \ac{CLTR} methods are negatively affected by trust bias, our affine approach approximates the optimal ranking model under varying degrees of position bias and trust bias.

The main contributions of this work are:
\begin{enumerate}[leftmargin=*,nosep]
	\item The first \ac{CLTR} estimator that is proven to be unbiased w.r.t. both position bias and trust bias.
	\item A theoretical analysis that shows that \ac{IPS} estimators cannot correct for trust bias.
	\item An empirical analysis based on semi-synthetic experiments that reveal our affine estimator bridges the gap between existing \ac{CLTR} methods and the optimal model when trust bias is present.
\end{enumerate}

\input{sections/tables/notations.tex}
\noindent%
Table~\ref{tab:notations} summarizes the notation we use in the paper. 

%% file: sections/tables/notations.tex
% !TEX root =  ../../2020-trust-bias.tex
\begin{table}
    \caption{Notation used in this paper.}
    \label{tab:notations}
    \begin{tabular}{l  l}
    \toprule
        \bf Symbol %& alias 
        & \bf Description \\
    \midrule
        $q$ & a query \\
        $d$ & an item (to be ranked) \\ 
        $D_q$ & set of items to be ranked for query $q$ \\
        $\lambda$ & metric function that assigns a weight per rank \\
        $f$ & a ranker, or ranking function, that scores items \\
        $y_i$ & a ranking displayed at interaction $i$ \\
        $C$ & a click on an item \\ %(random variable) \\
        $E$ & user examination of an item \\ %(random variable) \\
        $R$ & the relevance of an item \\ %(random variable) \\
        $\tilde{R}$ & the perceived relevance of an item \\ %(random variable)\\
        $\theta_k$ & examination probability at rank $k$: $P(E = 1 \mid k)$ \\
        $\gamma_{q,d}$ & relevance probability: $P(R = 1 \mid q, d)$ \\
        \multirow[t]{2}{*}{$\epsilon^+_k$} & perceived relevance probability at rank $k$ of \\
        & an examined relevant item: $P(\tilde{R}=1 \mid E = 1, R = 1, k)$ \\ 
       \multirow[t]{3}{*}{$\epsilon^-_k$} & perceived relevance probability at rank $k$ of \\
        & an examined non-relevant item:\\
        & $P(\tilde{R}=1 \mid E = 1, R = 0, k)$\\ 
        $\alpha_k$ & first weight of the affine transformation of \\ & trust bias: $\alpha_k = \theta_k(\epsilon^+_k - \epsilon^-_k)$ \\
        $\beta_k$ & second weight of the affine transformation of \\ &  trust bias: $\beta_k = \theta_k\epsilon^-_k$ \\
    \bottomrule
    \end{tabular}
\end{table}

%% file: sections/02-Background.tex
\vspace*{-0.59mm}
\section{Background and Related Work}

This section covers supervised \ac{LTR} and the original \ac{IPS} method for \ac{CLTR} with position bias correction.
\vspace*{-1.1mm}
\subsection{Learning to rank}

In general, the goal of \ac{LTR} methods is to find the optimal ranking function $f$, in order to sort items for user-issued queries.
For this work, we will use $f$ to sort in ascending order.
Let $q$ indicate a query, $d$ an item, and $\text{rank}(d \mid q, f)$ the rank of item $d$ in the ranking produced by $f$ for $q$. 
Then:
\begin{equation}
f(d_i \mid q) > f(d_j \mid q) \Rightarrow \text{rank}(d_i \mid q, f) \succ \text{rank}(d_j \mid q, f).
\end{equation}
Commonly, $f$ is considered optimal if it maximizes some linearly decomposable metric.
Let $P(q)$ be the distribution of queries, $D_q$ the set of items to be ranked for query $q$, and $P(R=1 \mid q,d)$ the probability that an item $d$ is considered relevant by the user. Then, with some weighting function $\lambda$, a linearly decomposable metric has the form:
\begin{equation}
    \label{eq:LTRreward}
    \Delta(f) = \sum_{q} P(q) \sum_{d \in D_q} P(R=1 \mid q,d) \cdot \lambda(d \mid q, f).
\end{equation}
Generally, $\lambda$ is based on the rank of $d$ for $q$ according to $f$. 
For instance, it can be chosen to match the well-known \ac{DCG} metric:
\begin{equation}
 \lambda_{\text{DCG}}(d \mid q, f) = \Big(\log_2\big(\text{rank}(d \mid D_q, q, f) + 1\big)\Big)^{-1}.
 \label{eq:DCG}
\end{equation}
If the relevance probabilities $P(R=1 \mid q,d)$ are known, finding the optimal $f$ can be done through traditional supervised \ac{LTR} methods~\citep{liu2009learning}.

\subsection{Counterfactual learning to rank for position bias correction}
\label{sec:backgroundCLTR}
In practice the relevance probabilities $P(R=1 \mid q,d)$ are not known and are costly to estimate through human labelling~\citep{qin2013introducing, Chapelle2011}.
Moreover, often the annotations obtained through manual labelling are not aligned with the actual preferences of the users~\citep{sanderson2010}.

As an alternative, \acf{CLTR} methods use click logs to base their optimization and evaluation on.
Clicks can be seen as a form of implicit feedback, which is indicative of the users' preferences but also a very noisy and biased signal.
One of the most prevalent biases in clicks on items included in a ranking is \emph{position bias}: users are less likely to examine -- and therefore click -- items on lower ranks.
Position bias is formally modeled through the examination hypothesis, which states that a clicked item ($C\in\{0,1\}$) must be examined ($E\in\{0,1\}$) and considered relevant ($R\in\{0,1\}$): $C = 1 \leftrightarrow E = 1 \land R = 1$.
Position bias is often assumed to depend on the rank at which an item is displayed, while the relevance of an item is assumed to be independent of where it is displayed~\citep{craswell2008experimental}.
Thus, if $k$ is the rank at which $d$ is displayed, the probability of a click is:
\begin{equation}
    P\big(C=1 \mid q,d,k\big) = P\big(E=1 \mid k\big) \cdot P\big(R=1 \mid q,d\big).
    \label{eq:positionbiasclicks}
\end{equation}
The click probability (Eq.~\ref{eq:positionbiasclicks}) shows us that the position bias, modeled by $P(E=1 \mid k)$, gives an unfair advantage to documents in positions that are more likely to be examined.

Let $\mathcal{D}$ be the set of logged interactions, containing $N$ tuples each consisting of a user-issued query $q_i$, a displayed ranking $y_i$, and the observed clicks $c_i$ where $c_i(d) \in \{0,1\}$:
\begin{equation}
\mathcal{D} = \{(q_i, y_i, c_i)\}_{i=1}^N.
\end{equation}
For brevity we will use the sum $\sum_{(d, k) \in y_i}$, which sums over the items $d$ and their associated ranks in $y_i$:
\begin{equation}
\forall (d, k) \in y_i, \,\, k = \text{rank}(d \mid i).
\end{equation}
Furthermore, we use $P(E=1 \mid k) = \theta_k$. 
Thus, the probability of item $d$ being examined at interaction $i$ depends on the rank it was displayed at:
$P(E=1 \mid d, i) = \theta_{\text{rank}(d \mid i)}$.
The first published \ac{CLTR} methods correct for position bias using an \ac{IPS} estimator~\citep{joachims2017unbiased, wang2016learning}.
This \ac{IPS} estimator weights each click inversely to the probability that the clicked item was examined:
\begin{equation}
    \hat{\Delta}_\mathit{IPS}(f) =
    \frac{1}{N} \sum_{i=1}^N \sum_{(d, k) \in y_i} \frac{c_i(d)}{\theta_k} \cdot \lambda(d \mid q_i, f).
    \label{eqn:M_IPS}
\end{equation}
The result is an unbiased estimator, since in expectation it correctly estimates $\Delta$:
\begin{equation}
    \mathbb{E}_{q,y,c}[\hat{\Delta}_\mathit{IPS}(f)] =  \Delta(f) \quad \text{(under the click model in Eq.~\ref{eq:positionbiasclicks})}.
\end{equation}
For a proof of unbiasedness we refer to the work by \citet{joachims2017unbiased}, who prove that even with click noise $\hat{\Delta}_\mathit{IPS}$ can be used for unbiased \ac{CLTR} optimization.
However, we note that this proof relies on (at least) three important assumptions:
\begin{enumerate*}[label=(\roman*)]
\item the click model as described in Eq.~\ref{eq:positionbiasclicks} is true,
\item the propensities $\theta$ are known, and
\item all propensities are positive: $\forall k\,\,  \theta_{k} > 0$.
\end{enumerate*}

A lot of related work has considered the estimation of the position bias parameters $\theta$, using randomization~\citep{joachims2017unbiased, fang2019intervention, agarwal2019estimating, wang2018position}, or by jointly estimating relevance and position bias~\citep{ai2018unbiased, wang2016learning}.
Recently, both \citet{ovaisi2020correcting} and \citet{oosterhuis2020topkrankings} have proposed using different propensities when not all items can be displayed at once (i.e., in case $\exists k\, \theta_k = 0$).
For this paper, we will assume that all propensities are positive and thus $\hat{\Delta}_\mathit{IPS}$ is unbiased.

Finally, different methods have been proposed to optimize $f$ based on $\hat{\Delta}_\mathit{IPS}$.
\citet{joachims2017unbiased} show that Rank-SVM~\citep{joachims2002optimizing} can be adapted to optimize \ac{IPS} estimates for the average-relevant-position metric.
\citet{agarwal2019counterfactual} introduce a method that can optimize any differentiable model w.r.t. an \ac{IPS} estimate of a metric based on a monotonically decreasing function.
Lastly, \citet{oosterhuis2020topkrankings} show that the supervised LambdaLoss \ac{LTR} framework~\citep{wang2018lambdaloss} can easily be adapted to optimize \ac{IPS} estimates as well.

%% file: sections/03-Trust.tex
\section{Trust Bias}
\label{sec:trust}
Besides position bias, other forms of bias are also known to affect user interactions with ranked lists.
\citet{joachims2005accurately} conclude that the trust users have in a ranking system affects their click behavior.
Because users trust the results, they are more likely to perceive top-ranked items to be relevant, even when the displayed information about the item suggests otherwise.
Similar to position bias, this causes items displayed at high ranks to have an unfair advantage, however, despite this similarity the effects of the two types of bias are not identical.

Recently, \citet{agarwal2019addressing} have modeled trust bias by distinguishing between \textit{perceived relevance} $\tilde{R} \in \{0, 1\}$ and \textit{real relevance} $R$.
Trust bias occurs because users are more likely to perceive items as relevant $\tilde{R} = 1$ if they are among the top ranked items in the list.
In \citeauthor{agarwal2019addressing}'s model, a click happens when a user examines and perceives an item to be relevant: $C = 1 \leftrightarrow E = 1 \land \tilde{R} = 1$.
The model combines rank-based position bias (as described in Section~\ref{sec:backgroundCLTR}) with trust bias, resulting in the following click probability:
\begin{equation}
    P\big(C=1  \mid q,d,k\big) = P\big(E=1 \mid k\big) \cdot P\big(\tilde{R}=1 \mid E = 1, R, k\big).
\end{equation}
Furthermore, the probability of perceived relevance of an examined item is conditioned on the actual relevance and the rank $k$ at which item $d$ is displayed.
For brevity,  we use $\epsilon^+_k$ and $\epsilon^-_k$ to denote these probabilities: 
 \begin{equation}
    \begin{split}
    P\big(\tilde{R}=1 \mid E = 1, R = 1, k\big) &=  \epsilon^+_k,
    \\
    P\big(\tilde{R}=1 \mid E = 1, R = 0, k\big) &=  \epsilon^-_k .
    \end{split}
    \label{eqn:epsilon}
\end{equation}
Additionally, we write $\gamma_{q,d}$ for the probability of actual relevance: $\gamma_{q,d} = P(R = 1 \mid q, d)$.
These conventions allow us to have the following succinct notation for the click probability:
 \begin{equation}
    P\big(C=1  \mid q,d,k\big) = \theta_k \left(\epsilon^+_k \gamma_{q,d} + \epsilon^-_k (1-\gamma_{q,d})\right).
    \label{eq:trustbiasclick}
\end{equation}
It is important to note that the combination of trust bias and position bias can be seen as an affine transformation between the relevance probabilities and click probabilities.
If we choose
$\alpha_k = \theta_k(\epsilon^+_k - \epsilon^-_k)$ and
$\beta_k = \theta_k\epsilon^-_k$,
this affine transformation becomes apparent:
\begin{equation}
    \label{eqn:affine_transformation}
    P\big(C=1  \mid q,d,k\big) = \alpha_k P(R = 1 \mid q, d) + \beta_k.
\end{equation}
We will use this property in Section~\ref{sec:method} to introduce affine corrections for these biases.

An empirical analysis by \citet{agarwal2019addressing} shows that their trust bias model better captures observed user behavior than the model that only considers position bias~(Eq.~\ref{eq:positionbiasclicks}).
Furthermore, \citeauthor{agarwal2019addressing} propose an \ac{IPS} estimator in order to correct for both trust bias and position bias.
In the next section, we will first prove that this estimator cannot correct for these biases. Moreover, we subsequently prove that no \ac{IPS} estimator is capable of doing so.
Then, in Section~\ref{sec:method} we introduce an estimator based on affine corrections, and prove that it is the first unbiased estimator that corrects for both position bias and trust bias.

%% file: sections/04-ExistingMethod.tex
\section{Existing Methods and Trust Bias}
\label{sec:existing}
In this section, we discuss \citet{agarwal2019addressing}'s Bayes-IPS method designed specifically for trust bias.
We prove that no \ac{IPS} estimator is able to correct for trust bias, including Bayes-IPS.

\subsection{Bayes-IPS}
% bayes-ips correction
\citet{agarwal2019addressing} have proposed the Bayes-IPS estimator to correct for trust bias and position bias.
This estimator combines two corrections:
\begin{enumerate*}[label=(\roman*)]
    \item correcting for position bias by weighting inversely to $\theta$; and
    \item correcting for trust bias by weighting each click to the probability of true relevance: $P(R=1 \mid \tilde{R}=1, E=1, k)$.
\end{enumerate*}
This results in the following estimator:
\begin{equation}
    \label{eqn:bayesips}
    \hat{\Delta}_{\BayesIPS}(f)
    =
    \frac{1}{N} \sum_{i=1}^N \sum_{(d,k) \in y_i}\frac{\epsilon^+_k}{\epsilon^+_k + \epsilon^-_k} \frac{c_i(d)}{\theta_k} \cdot \lambda(d \mid q_i, f).
\end{equation}
We note that $\hat{\Delta}_{\BayesIPS}$ is still an \ac{IPS} estimator; the difference with $\hat{\Delta}_\mathit{IPS}$ is that it uses the weights $\frac{1}{\theta_k}\frac{\epsilon^+_k}{\epsilon^+_k + \epsilon^-_k}$ instead of $\frac{1}{\theta_k}$.
In addition to $\theta_k$, Bayes-IPS also needs to know the values of $\epsilon^+_k$ and $\epsilon^-_k$.
\citeauthor{agarwal2019addressing} use \ac{EM} to estimate these values from click logs, and, using the estimated values, optimize a ranking model using $\hat{\Delta}_{\BayesIPS}$.
Their results show that optimizing with $\hat{\Delta}_{\BayesIPS}$ is more effective than using $\hat{\Delta}_{\mathit{IPS}}$ and leads to significant improvements when ranking for search through emails or other personal documents~\citep{agarwal2019addressing}.

While empirical results indicate that $\hat{\Delta}_{\BayesIPS}$ is an improvement over $\hat{\Delta}_{\mathit{IPS}}$, neither estimator is unbiased w.r.t.\ trust bias.
If trust bias is present, i.e., if $\exists k, k'\, (\epsilon^-_k \not=  \epsilon^-_{k'})$, then $\hat{\Delta}_{\BayesIPS}$ is biased.
%\ali{Furthermore, the bias may not be order preserving.}
% Harrie: Yes this is correct and a good point, however, I think the details on what we mean by biased should be discussed in 4.2., To me there seems little need to be so specific at this part of the text.
We can show this by looking at the difference between $\hat{\Delta}_{\BayesIPS}$ and $\Delta$, which is not necessarily equal to zero.
Let $\lambda_{q,d}$ be short for $\lambda(d\mid q,f)$, 
then:
\begin{align}
%    \begin{split}
   &\Delta(f) - \mathbb{E}_{q,y, c}\left[\hat{\Delta}_{\BayesIPS}(f)\right]
   \nonumber
    \\& =
    \mathbb{E}_{q, y,c}\!\left[\sum_{(d,k) \in y_i}\!
    \left(\gamma_{q,d} - \frac{\epsilon^+_k}{\epsilon^+_k + \epsilon^-_k} \frac{P(C = 1 \mid q, d, k)}{\theta_k} \right) \cdot \lambda_{q,d}\right]
%    \\& =
%     \mathbb{E}_{q, y} \left[\sum_{(d,k) \in y_i}
%    \left(\gamma_{q,d} - \frac{\epsilon^+_k}{\epsilon^+_k + \epsilon^-_k} \frac{
%     \theta_k \left(\epsilon^+_k \gamma_{q,d} + \epsilon^-_k (1-\gamma_{q,d})\right)
%    }{\theta_k} \right) \cdot \lambda(\cdot)\right]
     \\& =
     \mathbb{E}_{q, y,c}\!\left[\sum_{(d,k) \in y_i}\!
    \left(
    \left( 1 - \frac{\epsilon^+_k\left(\epsilon^+_k - \epsilon^-_k\right)}{\epsilon^+_k + \epsilon^-_k}
    \right)
    \gamma_{q,d} - \left( \frac{\epsilon^+_k\epsilon^-_k}{\epsilon^+_k + \epsilon^-_k}\right)
     \right) \cdot \lambda_{q,d}\right]
    .\nonumber
%    \end{split}
\end{align}
Clearly, it is non-trivial to derive under what conditions the difference between $ \Delta(f)$ and $\mathbb{E}_{q,y, c}[\hat{\Delta}_{\BayesIPS}(f)]$ is zero.
Instead of further investigating Bayes-IPS, we will prove that no \ac{IPS} estimator is unbiased w.r.t.\ trust bias under non-trivial circumstances, thereby also proving that no practical conditions exist where this difference is guaranteed to be zero.

\subsection{\ac{IPS} cannot correct for trust bias}
%Instead of further exploring the bias of Bayes-IPS, w
We proceed by considering whether any \ac{IPS} estimator can be unbiased w.r.t.\ trust bias.
Consider a generic \ac{IPS} estimator $\hat{\Delta}_{\rho}$. 
We will derive the values the propensities $\rho$ should have for unbiased \ac{CLTR}:
\label{sec:IPScannot}
\begin{equation}
\hat{\Delta}_{\rho}(f)
=
    \frac{1}{N} \sum_{i=1}^N \sum_{(d,k) \in y_i} \frac{c_i(d)}{\rho_{q_i,d,k}} \cdot \lambda(d \mid q_i, f).
    \label{eq:genericips}
\end{equation}
Importantly, we have to limit the possible choices for $\rho$, because trivially unbiased estimators are theoretically possible~\citep{hu2019unbiased}:
\begin{equation}
\forall d,k \,\, \rho_{q,d,k} =
    \frac{
    \frac{1}{N} \sum_{i=1}^N \sum_{(d,k) \in y_i} \gamma_{d,k} \cdot \lambda(d \mid q_i, f)
    }{
    \frac{1}{N} \sum_{i=1}^N \sum_{(d,k) \in y_i} c_i(d) \cdot \lambda(d \mid q_i, f)
    }.
\end{equation}
To avoid such trivial situations, we use the following definition for circumstances where \ac{CLTR} is not a trivial problem:

\begin{definition}
\label{def:non-trivialcircumstance}
We define \emph{non-trivial circumstances} as situations where no information about the relevances $\gamma$ is known.
Furthermore, trust bias must be present, meaning users' trust must not be constant at all the ranks:
\begin{equation}
\exists k, k' \,\, (\epsilon^-_{k} \not= \epsilon^-_{k'}).
\label{eq:definition:trustbias}
\end{equation}
Additionally, every displayed item should have a chance of being clicked and clicks at any rank $k$ should be positively correlated with relevance:
\begin{equation}
\forall k \,\, (\theta_k (\epsilon^+_{k} -\epsilon^-_{k}) > 0).
\label{eq:proof:relevancecorrelation}
\end{equation}
Lastly, the metric $\lambda$ should not be indifferent to the ranking of $f$:
\begin{equation}
\exists q, d, f, f' \,\, (\lambda(d \mid q, f) \not= \lambda(d \mid q, f')).
\end{equation}
\end{definition}

\noindent%
With this definition we avoid the following scenarios:
\begin{enumerate*}[label=(\roman*)]
\item $\rho$ is chosen based on the known values of $\gamma$, in which case there is no need to estimate $\Delta(f)$ based on clicks;
\item there is no trust bias, in which case every method is trivially unbiased w.r.t.\ trust bias;
\item some items cannot receive clicks or clicks are not indicative of relevance, in these cases there is no signal to learn from;
\item the metric is indifferent to the ranking function $f$, in which case there is nothing to evaluate since all ranking functions are equally good.
\end{enumerate*}

% To be clear about what we mean with \emph{unbiased w.r.t.\ trust bias} we introduce the following definition:
Naturally, an unbiased estimator should lead to the same optimal ranking as the full information case. For this, it is sufficient to have consistent pairwise rankings. To be clear about what we are going to prove about unbiasedness w.r.t. trust bias, we introduce the following formal definition:
\begin{definition}
\label{def:proof:requirement}
An \ac{IPS} estimator $\hat{\Delta}_\rho$ is \emph{unbiased w.r.t.\ trust bias}, if in all non-trivial circumstances $\rho$ can be chosen so that it can correctly predict relative differences:
\begin{equation}
 \exists \rho, \forall f, f',\,\,  
 \left(\Delta(f) >  \Delta(f') \leftrightarrow \mathbb{E}_c\left[\hat{\Delta}_{\rho}(f)\right] > \mathbb{E}_c\left[\hat{\Delta}_{\rho}(f')\right]\right).
\label{eq:proof:requirement}
   \end{equation}
\end{definition}

\noindent%
In other words, we define an estimator to be unbiased w.r.t.\ to trust bias, if it can unbiasedly predict the preference between any two rankers under any non-trivial circumstances.
Again, it is important to avoid $\rho$ being chosen based on knowledge of $\gamma$.
If a $\hat{\Delta}_\rho$ meets our definition of unbiasedness it can safely be applied in any non-trivial circumstances; we argue that this covers all realistic \ac{CLTR} situations.

\setcounter{theorem}{0}
\begin{theorem}
No \ac{IPS} estimator is unbiased w.r.t.\ trust bias. % (cf. Definition~\ref{def:proof:requirement}). % Harrie, it's true but the definition is the previous line!
\label{theorem:noips}
\end{theorem}

\begin{proof}
    We will prove this by showing that there are non-trivial circumstances where no values of $\rho$ exist for $\hat{\Delta}_\rho$ to correctly predict relative differences.
    We do so by starting from the most basic ranking example and deriving the values of $\rho$ where $\hat{\Delta}_\rho$ is unbiased, we prove that no such values exist.
    In addition to this proof, we will show how our basic example can be extended to include rankings with more queries and items.
    
    In our basic example, we consider a system that only receives a single query $q_1$: $P(q_1) = 1$ and that only has to rank two documents $D_{q_1} = \{d_1, d_2\}$.
    Therefore, two ranking functions can cover all possible rankings: $f_1$ that produces $[d_1, d_2]$ and $f_2$ that produces $[d_2, d_1]$.
    Lastly, the metric we consider is a top-1 metric, which means it is only affected by the top document of a ranking:
    \begin{equation}
    \begin{split}
    \lambda(d_1 \mid q_1, f_1) = \lambda(d_2 \mid q_1, f_2) > 0,\\
    \lambda(d_2 \mid q_1, f_1) = \lambda(d_1 \mid q_1, f_2) = 0.
    \end{split}
    \end{equation}
    Thus, in this basic example, we are trying to estimate whether $d_1$ should be ranked higher than $d_2$ or vice-versa.
     
    The true difference in metric value between the rankers is:
    \begin{equation}
        \Delta(f_1) - \Delta(f_2) = (\gamma_{q_1,d_1} - \gamma_{q_1,d_2}) \cdot \lambda(d_1 \mid q_1, f_1).
    \end{equation}
    Therefore, only the difference in item relevance matters for the relative difference:
    \begin{equation}
        \text{sign}\big(\Delta(f_1) - \Delta(f_2)\big) = \text{sign}(\gamma_{q_1,d_1} - \gamma_{q_1,d_2}).
    \label{eq:proof:reward}
    \end{equation}
    The estimates of $\hat{\Delta}_{\rho}$ are based on $N$ interactions with query $q_1$ where at each interaction $d_1$ and $d_2$ were displayed at rank $1$ and $2$, respectively.
    The difference in the expected estimates (cf.\ Eq.~\ref{eq:trustbiasclick} and Eq.~\ref{eq:genericips}) is therefore:
     \begin{equation}
     \begin{split}
     \mbox{}\hspace*{-1mm}\mathbb{E}_c&\left[\hat{\Delta}_{\rho}(f_1)\right] - \mathbb{E}_c\left[\hat{\Delta}_{\rho}(f_2)\right]
     \\ &=
         \frac{\theta_{1} \left( \left(\epsilon^+_{1} -  \epsilon^-_{1} \right) \gamma_{q_1,d_1} + \epsilon^-_{1}\right)}{\rho_{q_1,d_1,1}}
         - \frac{\theta_{2} \left( \left(\epsilon^+_{2} -  \epsilon^-_{2} \right) \gamma_{q_1,d_2} + \epsilon^-_{2}\right)}{\rho_{q_1,d_2,2}}
         .
     \end{split}
    \label{eq:proof:estimate}
    \end{equation}
%    \begin{equation}
%    \begin{split}
%    \mathbb{E}_c&\left[\hat{\Delta}_{\rho}(f_1)\right] - \mathbb{E}_c\left[\hat{\Delta}_{\rho}(f_2)\right] =
%        \frac{\alpha_1 \gamma_{q_1,d_1} + \beta_1}{\rho_{q_1,d_1,1}}
%        - \frac{\alpha_2 \gamma_{q_1,d_2} + \beta_2}{\rho_{q_1,d_2,2}}
%    \end{split}
%    \label{eq:proof:estimate}
%    \end{equation}
    We note that this scenario falls under the definition of a non-trivial circumstance (Definition~\ref{def:non-trivialcircumstance}).
%     moreover, we consider it the most basic non-trivial scenario possible.
    
    In order to be unbiased, the values of the propensities $\rho$ must be chosen so that the requirement in Eq.~\ref{eq:proof:requirement} is met.
    Note that for two continuous functions to always have the same sign, they should agree on zero values.
    By combining Eq.~\ref{eq:proof:requirement} with Eq.~\ref{eq:proof:reward} and Eq.~\ref{eq:proof:estimate}, we can derive that $\rho$ must meet the following requirement:
%     \begin{equation}
%            \gamma_{q_1,d_1} = \gamma_{q_1,d_2} \leftrightarrow 
%        \frac{\rho_{q_1,d_1,1} }{\rho_{q_1,d_2,2}}
%         = \frac{
%            \alpha_1 \gamma_{q_1,d_1} + \beta_1
%         }{
%            \alpha_2 \gamma_{q_1,d_2} + \beta_2
%         }.
%    \label{eq:proof:equality}
%     \end{equation}
          \begin{equation}
            \gamma_{q_1,d_1} = \gamma_{q_1,d_2} \leftrightarrow 
        \frac{\rho_{q_1,d_1,1} }{\rho_{q_1,d_2,2}}
         = \frac{
            \theta_{1} \left( \left(\epsilon^+_{1} -  \epsilon^-_{1} \right) \gamma_{q_1,d_1} + \epsilon^-_{1}\right)
         }{
           \theta_{2} \left( \left(\epsilon^+_{2} -  \epsilon^-_{2} \right) \gamma_{q_1,d_2} + \epsilon^-_{2}\right)
         }.
    \label{eq:proof:equality}
     \end{equation}
    Under non-trivial circumstances, $\rho$ has to be chosen without knowledge of $\gamma$, therefore we must find a single value for each of $\rho_{q_1,d_1,1}$ and $\rho_{q_1,d_2,2}$ that meets this requirement for all possible values of $\gamma$.
    Combining this fact with the fact that $\gamma$ consists of probabilities, we can derive the following requirement from Eq.~\ref{eq:proof:equality}:
%     \begin{equation}
%     \forall x \in [0,1], \,
%        \frac{\rho_{q_1,d_1,1} }{\rho_{q_1,d_2,2}}
%         = \frac{
%            \alpha_1 x + \beta_1
%         }{
%            \alpha_2 x + \beta_2
%         }.
%         \label{eq:proof:equality2}
%     \end{equation}
          \begin{equation}
     \forall x \in [0,1] \,\,
     \left(
        \frac{\rho_{q_1,d_1,1} }{\rho_{q_1,d_2,2}}
         = \frac{
            \theta_{1} \left( \left(\epsilon^+_{1} -  \epsilon^-_{1} \right) x + \epsilon^-_{1}\right)
         }{
           \theta_{2} \left( \left(\epsilon^+_{2} -  \epsilon^-_{2} \right) x + \epsilon^-_{2}\right)
         }
         \right).
         \label{eq:proof:equality2}
     \end{equation}
From this we can directly derive the following requirement for the bias parameters $\epsilon$:
\begin{equation}
\forall x \in [0,1] \,\,
\left(
\frac{
\epsilon^+_{1}
}{
\epsilon^+_{2}
}
=
\frac{
\epsilon^-_{1}
}{
\epsilon^-_{2}
}
= \frac{
\left(\epsilon^+_{1} -  \epsilon^-_{1} \right) x + \epsilon^-_{1}
}{
\left(\epsilon^+_{2} -  \epsilon^-_{2} \right) x + \epsilon^-_{2}
}
\right)
.
\label{eq:proof:equality3}
\end{equation}
Thus, a solution for the propensities $\rho$ in Eq.~\ref{eq:proof:equality2} only exists if the trust bias parameters $\epsilon$ meet the requirement in Eq.~\ref{eq:proof:equality3}.
Solving for $\epsilon$ shows that the latter requirement can be simplified to:
    \begin{equation}
        \frac{\epsilon^+_1}{\epsilon^+_2} = \frac{\epsilon^-_1}{\epsilon^-_2}.
        \label{eq:proof:epsilonrequirement}
    \end{equation}
    % Therefore, only in very specific cases where trust bias adheres to Eq.~\ref{eq:proof:epsilonrequirement} do values of $\rho$ exist that can meet Eq.~\ref{eq:proof:equality}.
    % If we compare Eq.~\ref{eq:proof:epsilonrequirement} with our definition of non-trivial circumstances and its requirement for trust bias in Eq.~\ref{eq:definition:trustbias}, we see that there are infinitely many possible values for $\epsilon$ that meet Eq.~\ref{eq:definition:trustbias} but contradict Eq.~\ref{eq:proof:epsilonrequirement}.
    % For the sake of our example, let the users' trust bias be: $\epsilon^+_{1} = 1$, $\epsilon^-_{1} = 0.8 $, $ \epsilon^+_{2} = 0.98 $ and $ \epsilon^-_{2} = 0.25$;\footnote{These specific values are approximate values extracted from the results in~\citep{agarwal2019addressing}.} since this does not meet Eq.~\ref{eq:proof:epsilonrequirement}, we conclude that no values for $\rho$ exist that meet Eq.~\ref{eq:proof:equality}.
    Therefore, only in very specific cases where trust bias adheres to Eq.~\ref{eq:proof:epsilonrequirement} do values of $\rho$ exist that can meet Eq.~\ref{eq:proof:equality}. 
    This proves the theorem since we have provided input cases where no IPS is unbiased. 
    In fact, in non-trivial circumstances (Definition~\ref{def:non-trivialcircumstance} and Eq.~\ref{eq:definition:trustbias}), the probability of even being close to this regularity of Eq.~\ref{eq:proof:epsilonrequirement} is so low that in practice we can safely say that it never happens. 
    So, not only do non-trivial input cases exist where no IPS can be unbiased, but almost all of the time we are dealing with such cases.
    
    Therefore, we have proven that $\hat{\Delta}_{\rho}$ can never correctly predict the relative difference between $f_1$ and $f_2$ in this example under non-trivial circumstances.
    In conclusion, we have therefore proven that no \ac{IPS} estimator is unbiased w.r.t.\ trust bias, since there are examples where under non-trivial circumstances, no propensities $\rho$ can be chosen so that it unbiasedly infers the preference between two rankers.
\end{proof}

\noindent%
While the basic counterexample used in the proof of Theorem~\ref{theorem:noips} is enough for proving that \ac{IPS} estimators are biased w.r.t.\ trust bias, we note that it can easily be extended to cases with more queries and items.
For any number of queries and items and any item pair $d_3$ and $d_4$, there always exist two rankers $f_3$ and $f_4$ that agree on all item placements expect that they swap the ranks of $d_3$ and $d_4$. % with different relevances Harrie: !? in the proof we use the fact that their relevance might be the same!
Using a proof analogous to the above, one can prove that similar to Eq.~\ref{eq:proof:epsilonrequirement} $(\epsilon^+_{k_3}/\epsilon^+_{k_4}) = (\epsilon^-_{k_3}/\epsilon^-_{k_4})$, where $k_3$ and $k_4$ are the display ranks of $d_3$ and $d_4$, respectively.
This process can be repeated for other item pairs until the requirement $\forall k, k'\,\, ((\epsilon^+_k/\epsilon^+_{k'}) = (\epsilon^-_{k}/\epsilon^-_{k'}))$ is obtained.
Thus, one can prove this very restrictive requirement to the trust bias, that applies regardless of the number of queries and documents.
Only when the trust bias adheres to this requirement, is it possible that an \ac{IPS} estimator may be able to correctly infer relative differences.
This shows that \ac{IPS} is not a practical solution to trust bias.

In summary, we have proven that no \ac{IPS} estimator is unbiased w.r.t.\ trust bias without \emph{a priori} knowledge of the relevance $\gamma$, and thus is not applicable in any practical circumstances.
We have done so by taking a generic \ac{IPS} estimator and deriving the possible values for the propensities $\rho$ that would lead to unbiased results in the most basic ranking scenario.
The proof of Theorem~\ref{theorem:noips} shows that for most instances of trust bias such values do not exist.
Thus, none of the existing \ac{IPS} estimators can correct for trust bias or can be adapted to do so. 
For clarity, this includes:
the original \ac{CLTR} estimators~\citep{joachims2017unbiased, wang2016learning}; the dual learning algorithm by \citet{ai2018unbiased}; the \ac{IPS} with corrections for item-selection bias by \citet{ovaisi2020correcting}; the policy aware estimator \citep{oosterhuis2020topkrankings}; and the Bayes-IPS estimator \citep{agarwal2019addressing}.

The problem with \ac{IPS} appears to be that trust bias causes an affine transformation between relevance probabilities and click probabilities.
For a single query item  pair $q,d$ displayed at rank $k$, ideally a propensity $\rho_{q,d,k}$ exists so that:
\begin{equation}
\gamma_{q,d} = \frac{\alpha_k \gamma_{q,d} + \beta_k}{\rho_{q,d,k}}.
\end{equation}
Such a propensity does exist but it is dependent on $\gamma$:
\begin{equation}
\rho_{q,d,k} 
= \frac{\alpha_k \gamma_{q,d} + \beta_k}{\gamma_{q,d}}.
\end{equation}
If $\beta_k = 0$ (i.e., $\epsilon^-_k = 0$), the transformation becomes linear and $\rho$ becomes independent of $\gamma$: $\rho_{q,d,k} = \alpha_k$.
Thus, the core issue is that \ac{IPS} applies a linear transformation to observed clicks but a linear transformation cannot correct for the affine transformation caused by trust bias.
As a solution to this problem, we will introduce a novel estimator that applies affine corrections to clicks.

%% file: sections/05-Method.tex
\section{Affine Corrections for Trust Bias}
\label{sec:method}
Next, we introduce our novel affine estimator: the first method that is proven to correct for trust bias.
We also compare the affine estimator with the existing \ac{IPS} estimator, and introduce an  adaption of the \ac{EM} algorithm for estimating trust bias.

\subsection{The novel affine estimator}
\label{sec:novelestimator}
In Section~\ref{sec:trust} we described how trust bias can be seen as an affine transformation from relevance probabilities to click probabilities (see Eq.~\ref{eqn:affine_transformation}).
Subsequently, in Section~\ref{sec:IPScannot} we proved that \ac{IPS} estimators cannot correct for trust bias because \ac{IPS} can only apply linear transformations and no linear transformation can reverse the effect of an affine transformation (in non-trivial circumstances).

We now propose a novel estimator based on affine transformations to correct for both position bias and trust bias: the affine estimator.
The estimator works for any situation where click probabilities are based on an affine transformation of relevance probabilities:
% we have stated this earlier, but I think it is better to be clear here, it is the introduction of the most important contribution after all
\begin{equation}
    P\big(C=1  \mid q,d,k\big) = \alpha_k P(R = 1 \mid q, d) + \beta_k.
    \label{eq:alphabetanotation}
\end{equation}
This includes trust bias where
$\alpha_k = \theta_k(\epsilon^+_k - \epsilon^-_k)$ and
$\beta_k = \theta_k\epsilon^-_k$.
Spelled out in the notation of Eq.~\ref{eq:alphabetanotation} and in the trust bias notation, the affine estimator is:
\begin{equation}
    \label{eqn:affine_correction}
    \begin{split}
    \hat{\Delta}_{\text{affine}}(f)
    &=
    \frac{1}{N} \sum_{i=1}^N \sum_{(d,k) \in y_i} \frac{c_i(d) - \beta_k}{\alpha_k} \cdot \lambda(d \mid q_i, f)
    \\
    &=
    \frac{1}{N} \sum_{i=1}^N \sum_{(d,k) \in y_i} \frac{c_i(d) - \theta_k\epsilon^-_k}{\theta_k(\epsilon^+_k - \epsilon^-_k)} \cdot \lambda(d \mid q_i, f)
    .
    \end{split}
\end{equation}
We see that the affine estimator reweights clicks inversely to $\alpha_k$, which is somewhat similar to \ac{IPS}.
However, the salient difference is that $\hat{\Delta}_{\text{affine}}$ also penalizes items by subtracting $\frac{\beta_k}{\alpha_k}$.
This penalty compensates for \emph{incorrect} clicks where the perceived relevance does not match the true relevance: $\tilde{R} = 1 \land R = 0$.
Thus items displayed at ranks where more \emph{incorrect} clicks take place receive more penalties, while simultaneously clicks are reweighted according to the position bias $\theta_k$ and to compensate for the penalties: $\epsilon^+_k - \epsilon^-_k$.
We note that unlike with the \ac{IPS} estimator, an item that is displayed but not clicked can receive a negative weight.
In expectation later clicks will compensate for this effect.

\begin{theorem}
The affine estimator is unbiased w.r.t\ trust bias.
\end{theorem}

\begin{proof}
First, we use the assumption that clicks are correlated with relevancy: $\forall k\,\, (\alpha_k \not= 0)$.
Then we consider the expected value for a single click $c_i(d)$:
\begin{equation}
\mathbb{E}_{c}\left[  \frac{c_i(d) - \beta_k}{\alpha_k} \right]
=
\frac{(\alpha_k \cdot \gamma_{q,d} + \beta_k) - \beta_k}{\alpha_k}
=
\gamma_{q,d}.
\end{equation}
We can use this to derive the expected value of the \OurMethod estimator; it is equal to the true metric value:
\begin{equation}
\begin{split}
        \mathbb{E}_{q, y, c}&\left[\hat{\Delta}_{\text{\OurMethodShort}}(f)\right] 
        \\& =\mathbb{E}_{q, y} \left[\sum_{(d,k) \in y_i} \mathbb{E}_{c}\left[  \frac{c_i(d) - \beta_k}{\alpha_k} \right] \cdot \lambda(d \mid q, f)\right]
        \\& =\mathbb{E}_{q, y} \left[\sum_{(d,k) \in y_i} \gamma_{q,d} \cdot \lambda(d \mid q,f)\right] 
         = \Delta(f).
\end{split}
\end{equation}    
Therefore, the \OurMethod estimator is unbiased in expectation.
\end{proof}

\noindent%
The negative penalties $(\beta_k/\alpha_k)$ may be counter-intuitive. 
For a better understanding we consider a maximally non-relevant item $\gamma_{q,d} = 0$ that is displayed at rank $k$, $M$ times.
We expect to observe $M\cdot\beta_k$ clicks (all incorrect since $\gamma_{q,d} = 0$).
The sum of the penalties for the item given by the affine estimator is $M \cdot (\beta_k/\alpha_k)$, while each click is weighted by $1/\alpha_k$.
Thus, if we sum the weights for the clicks we expect $M \cdot (\beta_k/\alpha_k)$, therefore taking this sum minus the penalties correctly results in a zero weight for the item (in expectation).
As with any estimator, for reliable estimates, $M$ needs to be considerably large due to variance.

This concludes the introduction of our novel affine estimator.
By performing affine transformations to clicks it is the first estimator that can correct for the effect of both position bias and trust bias.

\subsection{Relation with \ac{IPS} and other properties}

While the \OurMethod estimator is very distinct from \ac{IPS} since it can perform corrections that \ac{IPS} cannot, we consider the former to be an extension of the latter.
In the most straightforward way any \ac{IPS}-based estimator can be seen as a special case of the \OurMethod estimator where $\forall k\,\, (\beta_k = 0)$.
More generally, we consider the situation without trust bias, i.e., where $\epsilon^+_k$ and $\epsilon^-_k$ have the same value for every $k$: $\forall k, k' \,\, (\epsilon^-_k=\epsilon^-_{k'}=\epsilon^- \land \epsilon^+_k=\epsilon^+_{k'}=\epsilon^+)$ and where clicks are positively correlated with relevance: $\epsilon^+ > \epsilon^-$.
Also, we will assume that summing $\lambda$ over documents leads to a constant value:
\begin{equation}
\forall f,  f', q \,\, \left(\sum_{d\in D_q} \lambda(d \mid q, f) = \sum_{d\in D_q} \lambda(d \mid q, f')\right).
\end{equation}
This means that if all items are equally relevant the order of the items does not matter.
We note that this holds for virtually all ranking metrics, e.g., \ac{DCG}, precision, recall, MAP, ARP, etc.
Now, Eq.~\ref{eqn:affine_correction} can be rewritten as follows:
\begin{equation}
    \label{eqn:affine_generalization_of_IPS}
    \begin{split}
    \hat{\Delta}&_\text{\OurMethodShort}(f) 
    \\& = \frac{1}{N} \frac{1}{\epsilon^+ - \epsilon^-} \sum_{i=1}^N \sum_{(d,k) \in y_i} \left(\frac{c_i(d)}{\theta_k} - \epsilon^- \right)\cdot \lambda(d \mid q_i, f) 
    \\& = \frac{1}{\epsilon^+ - \epsilon^-} \hat{\Delta}_\mathit{IPS}(f) 
    - \frac{1}{N} \frac{\epsilon^-}{\epsilon^+ - \epsilon^-} \sum_{i=1}^N \sum_{(d,k) \in y_i} \lambda(d\mid q_i, f)
    \\& = \frac{1}{\epsilon^+ - \epsilon^-} \hat{\Delta}_\mathit{IPS}(f) - \mathbb{C},
\end{split}
\end{equation}
\noindent
where $\mathbb{C}$ is a constant independent of $f$.
Therefore, $\hat{\Delta}_\text{\OurMethodShort}$ unbiasedly predicts relative differences w.r.t. $\hat{\Delta}_\mathit{IPS}$:
\begin{equation}
\forall f, f', \, \hat{\Delta}_\text{\OurMethodShort}(f) > \hat{\Delta}_\text{\OurMethodShort}(f')
\leftrightarrow
\hat{\Delta}_\mathit{IPS}(f) > \hat{\Delta}_\mathit{IPS}(f').
\end{equation}
Consequently, we can conclude that optimizing $f$ w.r.t. $\hat{\Delta}_\text{\OurMethodShort}(f)$ also optimizes w.r.t.\ $\hat{\Delta}_\mathit{IPS}$ when trust bias is not present.
This further shows that the \OurMethod estimator should be viewed as a generalization of the existing \ac{IPS} approach.

Furthermore, the notation of the \OurMethod estimator in Eq.~\ref{eqn:affine_correction} also reveals some other intuitive properties.
We see that if for some $k$, $\alpha_k = 0$, then the estimator becomes undefined, thus if clicks are not correlated with relevance, the estimator cannot be applied.
Interestingly, if we compare this with the trust bias model we see that there are only two cases when $\exists k\, (\alpha_k = 0)$ can occur:
\begin{enumerate*}[label=(\roman*)]
\item when $\exists k\, (\theta_k = 0)$, i.e., at some rank $k$ some items cannot be observed or clicked, hence nothing about the item at this rank can be learned; or 
\item when $\exists k\,  (\epsilon^+_k = \epsilon^-_k)$, i.e., at some rank $k$ non-relevant and relevant items are equally likely to be clicked, thus there is nothing to learn from the click signal.
\end{enumerate*}
Furthermore, something interesting happens if $\exists k\,  (\epsilon^+_k < \epsilon^-_k)$, i.e., if at some rank $k$ \emph{non-relevant} items are \emph{more} likely to be clicked than \emph{relevant} items.
In this case, non-clicked items receive a positive penalty and clicks lead to negative scores, meaning the less clicked items are preferred since they are more likely to be relevant.
All these cases are very intuitive and we consider it a great strength that they can be inferred from the \OurMethod estimator from just a brief analysis of its formulation.

\subsection{Parameter estimation}
\label{subsec:parametersestimation}
\citet{agarwal2019addressing} describe how \ac{EM} can be used to estimate the position bias and trust bias parameters.
We also use the regression-based \ac{EM} procedure for estimating the bias parameters.
However, unlike \citeauthor{agarwal2019addressing}, who estimate three parameters per rank $k$, namely $\theta_k$, $\epsilon^-_k$ and $\epsilon^+_k$, we notice that only two have to be estimated:
\begin{equation}
\begin{split}
\zeta^+_k &= P(C=1 \mid R=1, k) = \theta_k \epsilon^+_k=\alpha_k + \beta_k\\
\zeta^-_k &= P(C=1 \mid R=0, k) = \theta_k \epsilon^-_k=\beta_k.
\end{split}
\end{equation}
From these two parameters the value of $\alpha_k$ and $\beta_k$ can be inferred directly ($\alpha_k = \zeta^+_k - \zeta^-_k$), and these are the only parameters required for the trust bias click model (Eq.~\ref{eqn:affine_transformation}) and the \OurMethod estimator (Eq.~\ref{eqn:affine_correction}).

To estimate these parameters we adapt the Expectation step, where the parameters are updated as follows:
\begin{equation}
\begin{split}
    &\zeta^+_k = \\[-2pt]
     &\frac{\sum_{i=1}^N c_i(d)P(R=1 | C=1, q_i, d, k)}{\sum_{i=1}^N c_i(d) P(R=1 | C=1, \mydots) + (1 - c_i(d))P(R=1 | C=0, \mydots)},
\end{split}     
\end{equation}
and
\begin{equation}
\begin{split}
        & \zeta^-_k = \\[-2pt]
         &\frac{\sum_{i=1}^N c_i(d)P(R=0 \mid C=1, q_i, d, k)}{\sum_{i=1}^N c_i(d) P(R=0 | C=1, \mydots) + (1 - c_i(d))P(R=0 | C=0, \mydots)},
\end{split}
\end{equation}
\noindent
where the conditional relevance probabilities $P(R \mid C, q, d, k)$ are computed using Bayes's law.
This simplification allows us to estimate the parameters with less computational costs.
And since fewer parameters are estimated, we expect \ac{EM} to converge faster.

% Harrie: I think this should be discussed in experimental setup, since it's not something novel about our algorithm.
In the Maximization step, the $\gamma$ values are estimated by a regression algorithm.
We use $\zeta^-$ and $\zeta^+$ obtained from the E-step to train the unbiased ranking function $f$ based on $\hat{\Delta}_\text{{\OurMethodShort}}$.
Previous work~\citep{agarwal2019addressing,wang2018position} suggests to use a \emph{sigmoid} as a final activation function to obtain valid probability values.
However, we observed that the sigmoid function gives very similar relevance probabilities between most items.
In contrast, the \emph{softmax} function results in more varied values but it forces the probabilities to sum to one for each query.
As a simple alternative we propose the \emph{soft-min-max} function, which does not force probabilities to sum to one, but still results in varied values:
\begin{equation}
    \text{soft-min-max}(x_i)=\frac{e^{x_i} - e^{\min(x_i)}}{e^{\max(x_i)} - e^{\min(x_i)}}.
\end{equation}
Our experiments show that the choice of activation function leads to noticeable differences.

%% file: sections/06-Experiments.tex
% !TEX root = ../2020-trust-bias.tex

\section{Experimental Setup}
%This section details the experiments we ran to evaluate the \OurMethod estimator.
%
%\subsection{General Semi-Synthetic Setup}
We follow the semi-synthetic setup that is prevalent in existing \ac{CLTR} work~\citep{joachims2017unbiased, ai2018unbiased, oosterhuis2020topkrankings, jagerman2019comparison}, where queries, documents and relevances are sampled from supervised \ac{LTR} datasets, while clicks are simulated using probabilistic user models.
First, we train a production ranker for each dataset; we randomly select 20 queries from each training set and use the supervised \ac{LTR} LambdaMART method to optimize a ranking model.
With these production rankers we simulate a situation where a decent ranking system exists but still leaves plenty of room for improvement.
On each dataset, we simulate user interactions by repeatedly: 
\begin{enumerate*}[label=(\roman*)]
\item uniform-random sampling a query from the training set, 
\item ranking the documents for that query with the production ranker, and 
\item simulating clicks on the resulting ranking using a probabilistic user model.
\end{enumerate*}
This semi-synthetic setup allows us to vary the number of clicks available for learning, as well as the position bias and trust bias of the simulated user.
Thus, we can analyze the effects these factors have on the \OurMethod estimator and other \ac{CLTR} methods.

% \vspace*{-1.6mm}
\subsection{Datasets}
We use two of the largest publicly available \ac{LTR} datasets: Yahoo! Webscope~\citep{Chapelle2011} and MSLR-WEB30k~\citep{qin2013introducing}.
Both were created by a commercial search engine, and each contains around \numprint{30000} queries, each query has a set of preselected documents to be ranked.
The datasets contain five level relevancy tags acquired through expert labelling for the preselected query-document pairs.
Yahoo!\ has 24 documents per query on average and uses 700-feature vectors to represent query-documents; MSLR has 125 per query and uses 136 features.
Each dataset is split in training, validation and test sets; we only use the first fold of MSLR.

% \vspace*{-1.6mm}
\subsection{Click simulation}
Clicks are simulated on rankings produced by the production rankers by applying probabilistic click models.

Per experimental setting, we simulate up to $8 \cdot 10^6$ clicks on the training set.
The number of validation clicks is always 15\% and 33\% of the training clicks for Yahoo! and MSLR, respectively.
These numbers were chosen to match the ratio between the number of training and validation queries in each dataset.

We apply \citet{agarwal2019addressing}'s trust bias model with varying parameters (see Section~\ref{sec:trust}).
The relevances $\gamma_{q,d}$ are based on the relevance label recorded in the datasets; we follow \citet{joachims2017unbiased} and use binary relevance:
\begin{equation}
P(R = 1 \mid q, d) = \gamma_{q,d} =
    \begin{cases}
    
    1 & \text{if relevance\_label}(q,d) > 2,
    \\
    0 & \text{otherwise}.
    \end{cases}
\end{equation}
Similar to previous work~\citep{joachims2017unbiased, oosterhuis2020topkrankings, jagerman2019comparison}, we set the position bias inversely proportional to the display rank:
\begin{equation}
P(E = 1 \mid k) = \theta_k = \left(\frac{1}{\min(k, 20)}\right)^\eta,
\end{equation}
where we vary the $\eta$ parameter: $\eta \in \{1, 2\}$.

To the best of our knowledge, this is the first \ac{CLTR} that simulates trust bias, thus there is no precedent for the values of $\epsilon^+_k$ and $\epsilon^-_k$.
In order to simulate trust bias as realistically as possible, we base our values on the empirical work of \citet{agarwal2019addressing}.
It appears that the bias \citeauthor{agarwal2019addressing} inferred from actual user interactions can be approximated by the following formula:
\begin{equation}
\forall k \in \{1, 2, \ldots, 5\}, \quad   
    \epsilon^+_k \approx 1 - \frac{k+1}{100} \quad \land
\quad
    \epsilon^-_k \approx \epsilon^-_1 \frac{1}{k}.
\end{equation}
Unfortunately, \citeauthor{agarwal2019addressing} only observed interactions on top-5 rankings.
To prevent $\epsilon^+_k$ and $\epsilon^-_k$ from disappearing on ranks beyond $k=5$, we apply the following
% \footnote{We are using different sentinels for the $\min$ functions because we want to test with higher values of simulated noise; i.e. lower $\epsilon^+_k$ and higher $\epsilon^-_k$.}: 
\begin{equation}
 \epsilon^+_k = 1 - \frac{\min(k, 20)+1}{100},
\quad
    \epsilon^-_k = \epsilon^-_1 \frac{1}{\min(k, 10)}.
\end{equation}
We use the \emph{incorrect-click} rate on the first rank: $\epsilon^-_1$, as a hyper-parameter to vary the amount of trust bias.
We found that our results are consistent across different values for $\epsilon^-_1$.
To cover both cases with high and low trust bias, we report results with $\epsilon^-_1 \in \{0.65, 0.35\}$.

\vspace*{-2mm}
\subsection{\ac{LTR} algorithm}
Similar to \citet{ai2018unbiased} and \citet{agarwal2019counterfactual} we train neural networks for our ranking functions.
Our preliminary results indicate that the configuration of the networks does not have to be fine-tuned.
The reported results are produced using models with three hidden layers with sizes $\left[512,256,128\right]$ respectively.
All layers use $elu$ activations and $0.1$ dropout was applied to the last two layers.

For the loss function we follow \citet{oosterhuis2020topkrankings} and use LambdaLoss to optimize \ac{DCG}~\citep{wang2018lambdaloss}.
For updating the gradients, we use the AdaGrad optimizer~\citep{duchi2011adaptive} with a learning rate of $0.004$ and $0.02$ for Yahoo! and MSLR datasets respectively, for 32 epochs.
% Finally, we perform early stopping based on an estimate of \ac{DCG} on the validation clicks.
%We notice that early stopping is not necessary when using dropout.

\vspace*{-2mm}
\subsection{Experimental runs}
\label{subsec:experimentalruns}

We evaluate the performance of our~\OurMethod estimator, by comparing the nDCG@10 of the models it produces with those produced using other estimators.
The following estimators are used as baselines:
\begin{enumerate}[leftmargin=*]%[label=(\roman*)]
    \item \textbf{No Correction}: The na\"ive estimator where each click is treated as an unbiased relevance signal.
    \item \textbf{IPS}: The original \ac{CLTR} \ac{IPS} estimator~\citep{joachims2017unbiased,wang2018position} that only corrects for position bias (see Section~\ref{sec:backgroundCLTR}).
     \item \textbf{Bayes-IPS}: The only existing \ac{CLTR} estimator~\citep{agarwal2019addressing} designed for addressing trust bias (see Section~\ref{sec:existing}).
\end{enumerate}
For a clearer analysis, we also report the performance of the following ranking models:
\begin{enumerate}[leftmargin=*,resume]%[label=(\roman*),resume]
    \item \textbf{Production}: The production ranker used in during the logging of simulated clicks.
    \item \textbf{Full Info}: A model trained using supervised \ac{LTR} on the true relevance probabilities, its performance illustrates the (theoretical) maximal performance possible on a dataset.
    We note that this is not a baseline as it does not learn from clicks but (unrealistically) from the true relevances. 
\end{enumerate}
All reported nDCG@10 results are an average of four independent runs.
Our experiments cover both the situation where the bias ($\theta$, $\epsilon^-$ and $\epsilon^+$) is known, e.g., through previous experiments~\citep{fang2019intervention, agarwal2019estimating, wang2018position}, and the situation where the bias has to be estimated still.

%% file: sections/07-Results.tex
\setlength{\tabcolsep}{0.15em}

{\renewcommand{\arraystretch}{0.01}
\begin{figure*}[t]
\centering
\begin{tabular}{l c c c c}
&
\small $\eta = 1$ and $\epsilon^-_1 = 0.65$
&
\small $\eta = 2$ and $\epsilon^-_1 = 0.65$
&
\small $\eta = 1$ and $\epsilon^-_1 = 0.35$
&
\small $\eta = 2$ and $\epsilon^-_1 = 0.35$
\\
\rotatebox[origin=lt]{90}{\hspace{2.25em} \small nDCG@10} &
\includegraphics[scale=0.3]{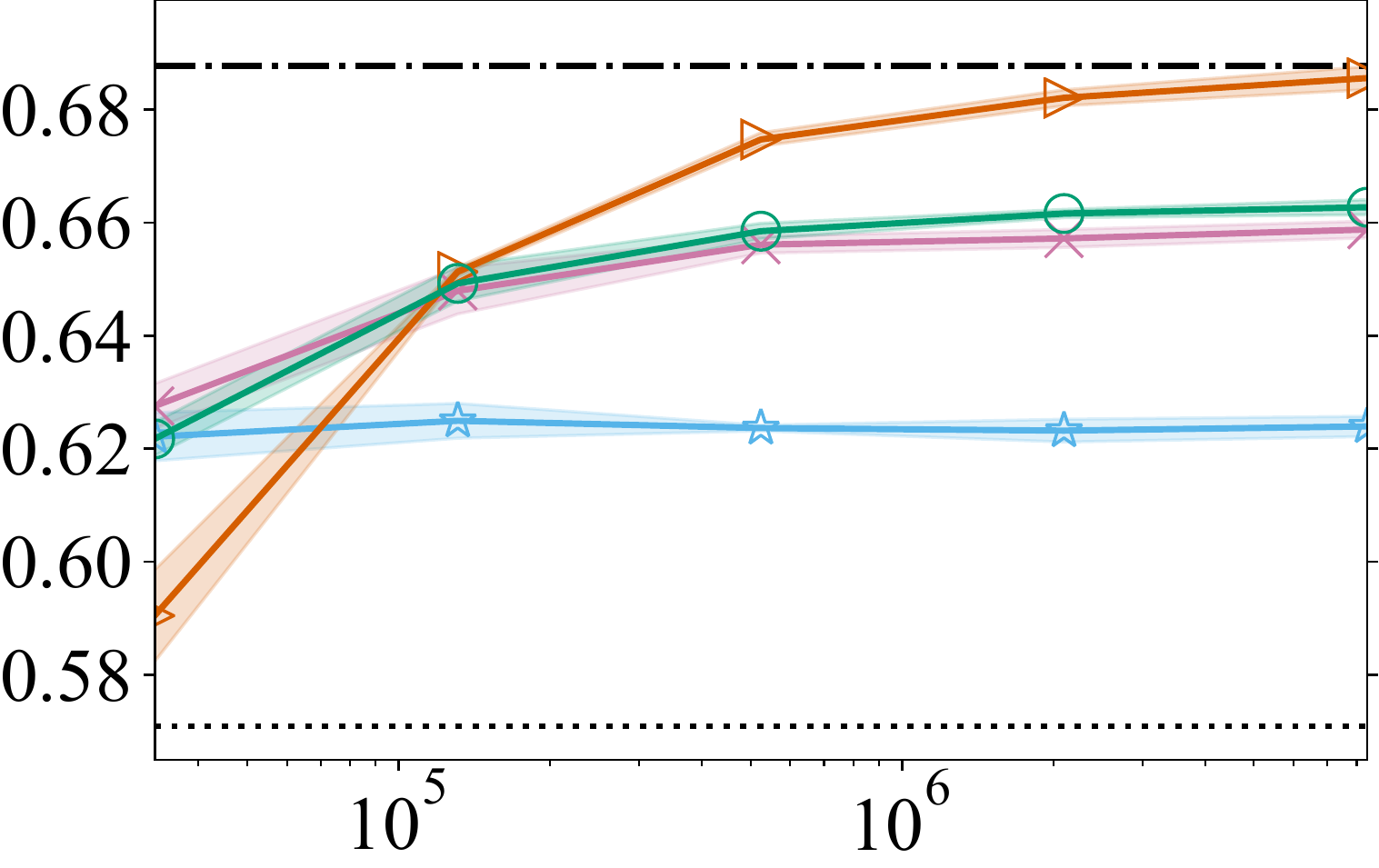} &
\includegraphics[scale=0.3]{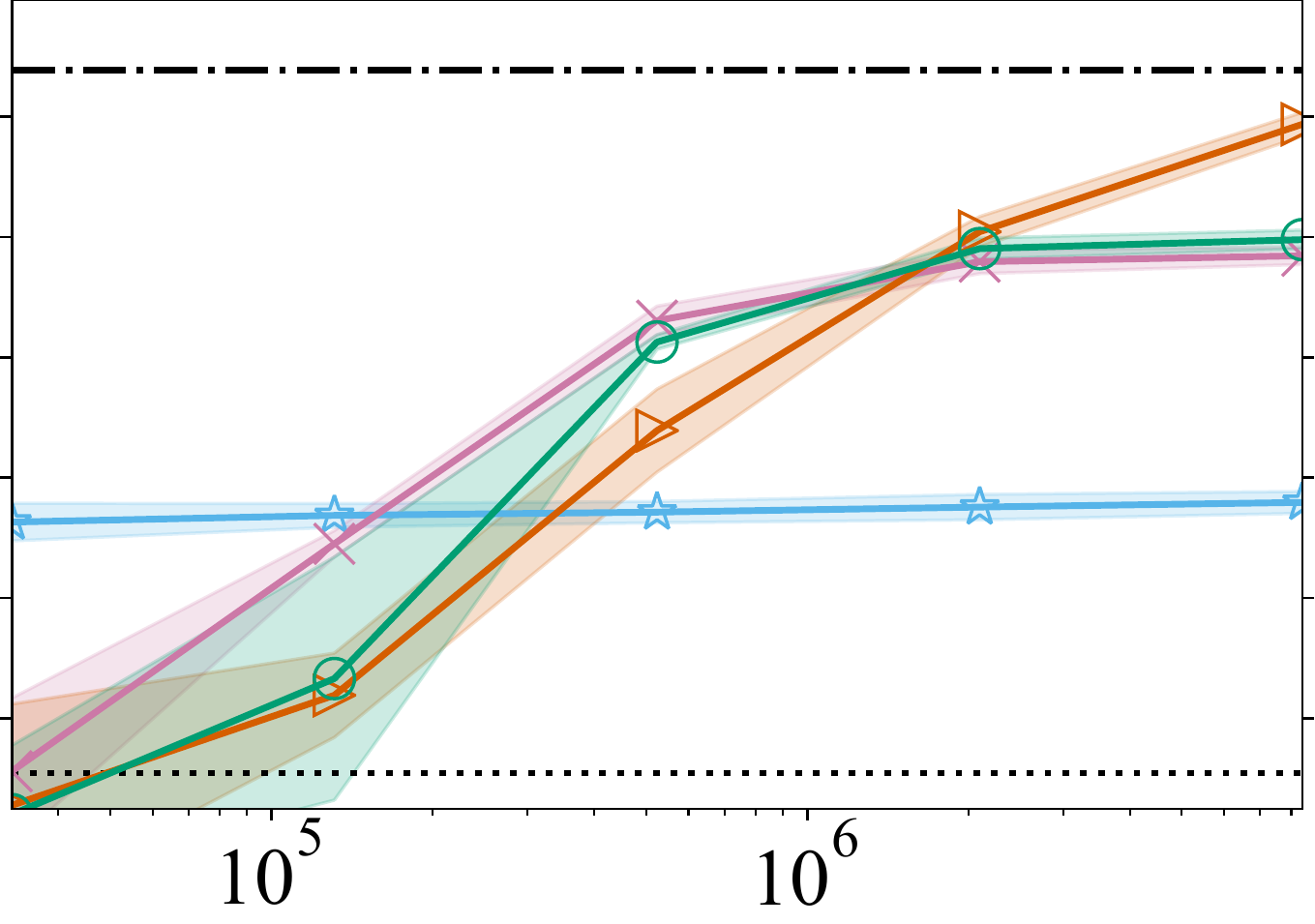} &
\includegraphics[scale=0.3]{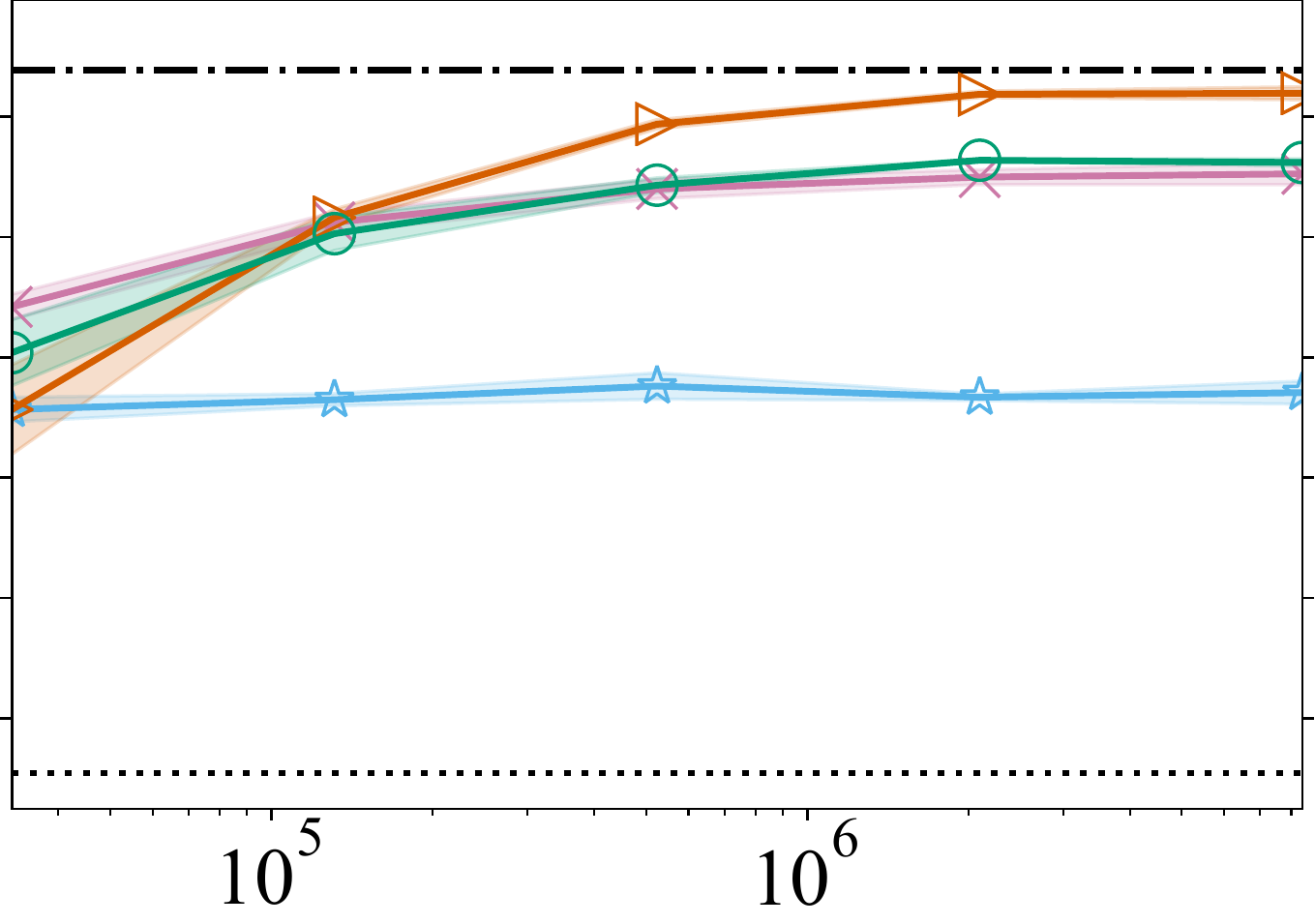} &
\includegraphics[scale=0.3]{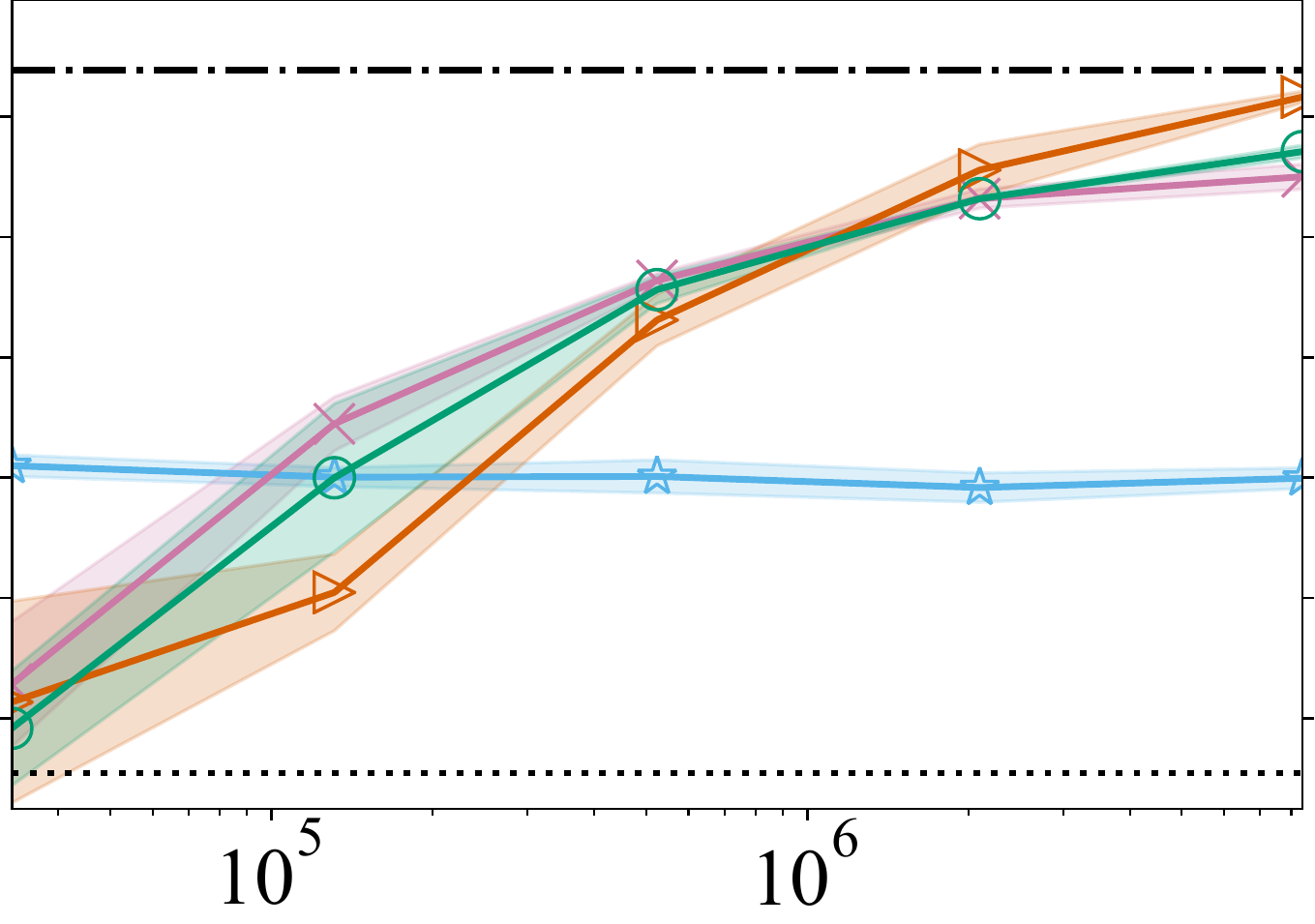}
\\
\rotatebox[origin=lt]{90}{\hspace{2.4em} \small nDCG@10}  &
\includegraphics[scale=0.3]{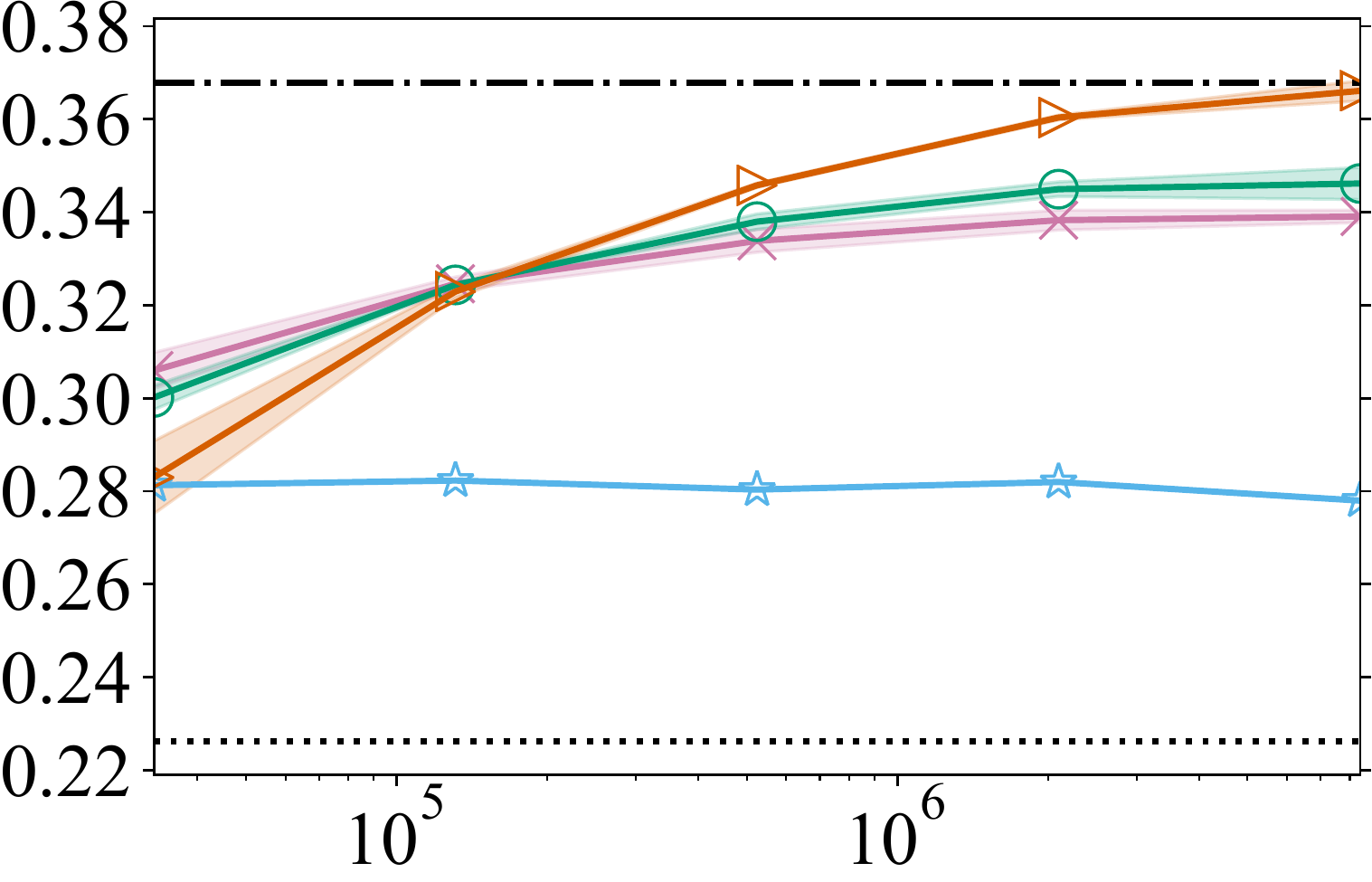} &
\includegraphics[scale=0.3]{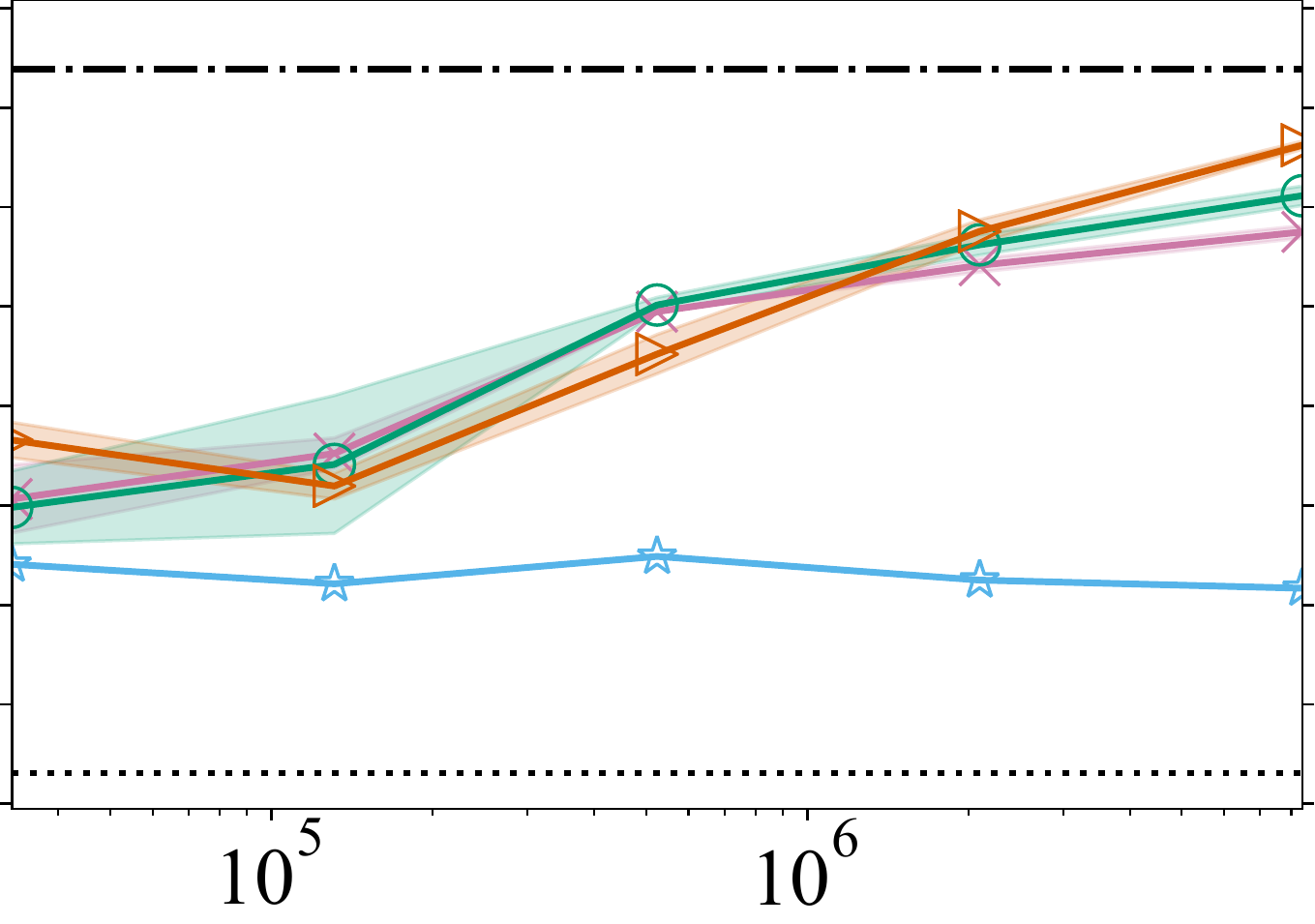} &
\includegraphics[scale=0.3]{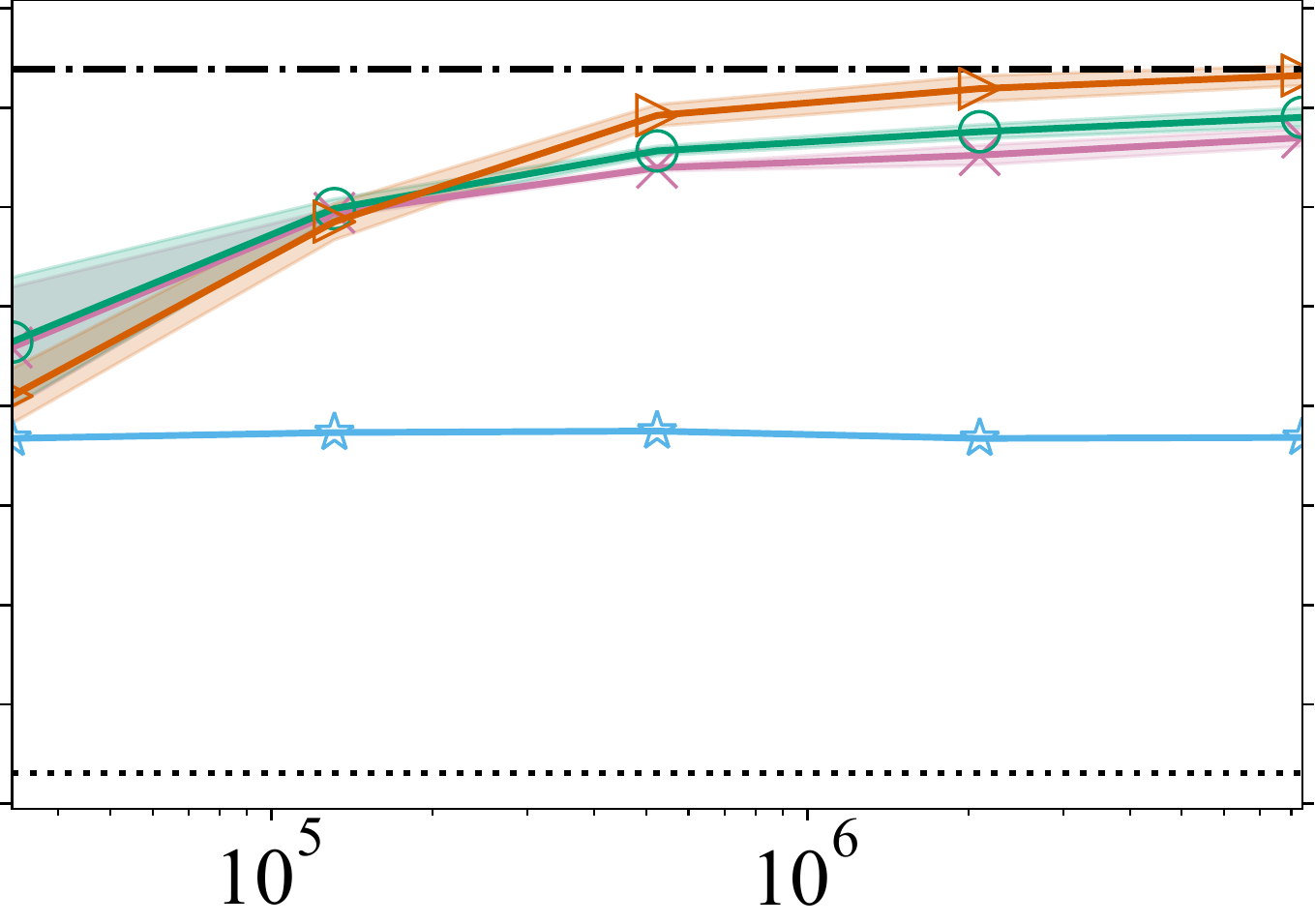} &
\includegraphics[scale=0.3]{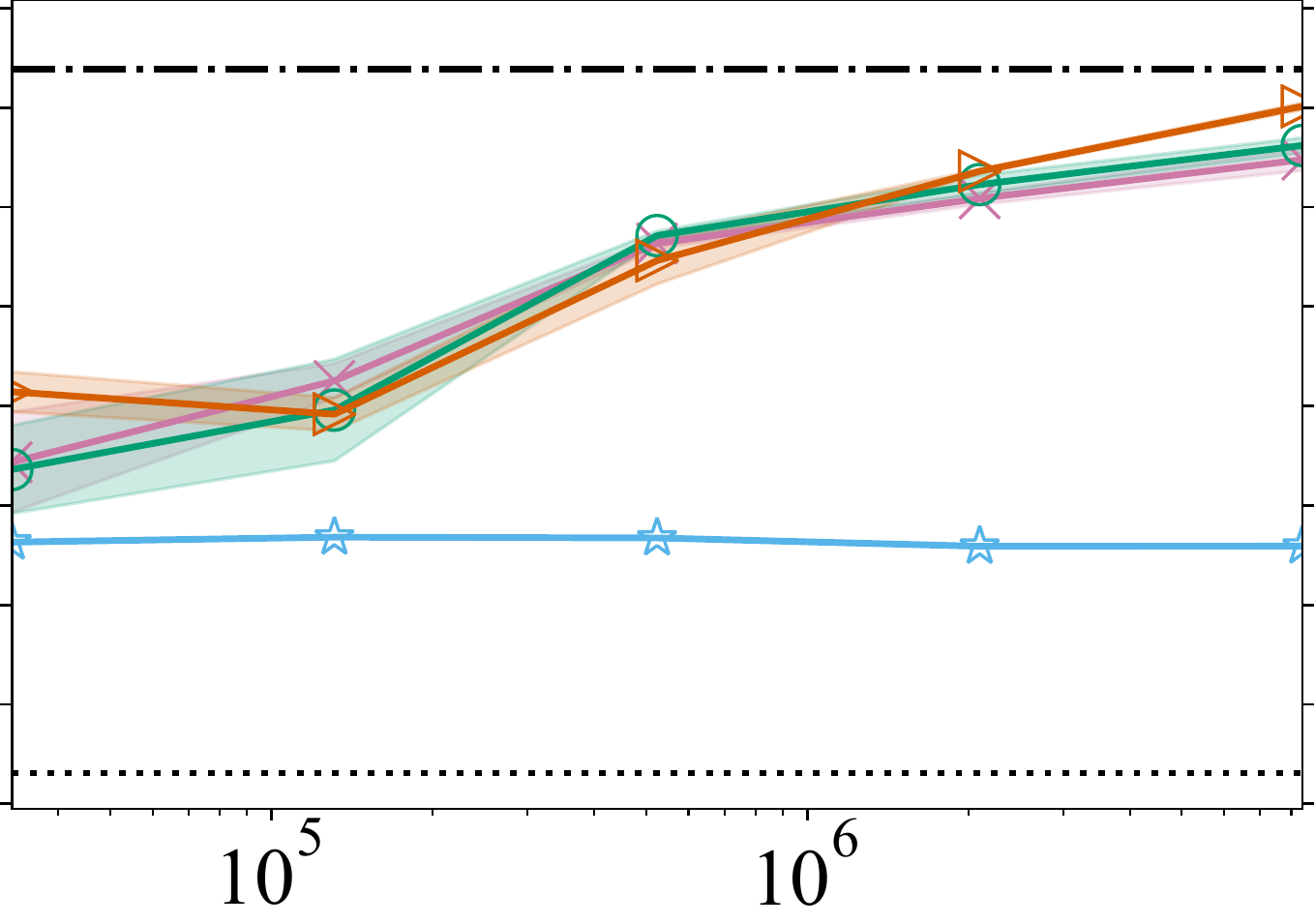}
\\
& \small Number of Training Clicks
& \small Number of Training Clicks
& \small Number of Training Clicks
& \small Number of Training Clicks
\\
 \multicolumn{5}{c}{
 \includegraphics[scale=.47]{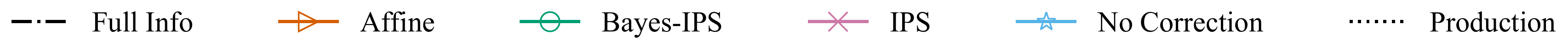}
} 
\end{tabular}
% \vspace{-0.5\baselineskip}
\caption{
Comparison of different \ac{CLTR} estimators in term of nDCG@10 on different numbers of clicks and under varying levels of position bias and trust bias.
Estimators were given the true bias parameters.
Results are averaged over four runs; shaded area indicates the standard deviation. 
Top row: Yahoo! Webscope dataset; bottom row: MSLR-WEB30k dataset.
}
\label{fig:performance}
\end{figure*}
}

\vspace*{-1mm}
\section{Results and Discussion}
This section discusses our experimental results. 
We consider the ranking performance of the affine estimator compared to other estimators, in both the situation where the exact bias is known and where it has to be estimated.

\vspace*{-1mm}
\subsection{Optimization with the affine estimator}
\label{subsec:performance}

First we consider whether \emph{optimizing with the affine estimator leads to better performing ranking models than with existing estimators}.

Figure~\ref{fig:performance} shows the performance (nDCG@10) reached by the different estimators under varying degrees of bias and different numbers of clicks available for training.
We see that the na\"ive estimator has already converged after $3\cdot10^5$ clicks, since additional clicks do not increase its performance.
In line with the empirical results of \citet{agarwal2019addressing}, we see that both \ac{IPS} and Bayes-IPS improve over the na\"ive estimator, and that Bayes-IPS consistently outperforms \ac{IPS}.
However, when we compare with the Full Info ranker, we see that there is still a sizable gap between Full Info and Bayes-IPS in every tested setting on both datasets.
In other words, neither IPS nor Bayes-IPS can approximate the optimal model under the tested degrees of trust bias.
As predicted by the theory in Section~\ref{sec:existing}, it thus appears that both these \ac{IPS} estimators are biased w.r.t. trust bias.

In contrast, we see that the affine estimator does approximate the optimal model when position bias is mild ($\eta = 1$).
However, under extreme position bias ($\eta = 2$) it has not reached convergence in any of our graphs.
Based on the theory in Section~\ref{sec:novelestimator}, we expect convergence near the optimal model if it were given more training clicks.
Furthermore, in all tested settings we observe the affine estimator to outperform the other estimators when more than $10^6$ training clicks are available.
Using the Student's t-test we found that all the improvements at $8 \cdot 10^6$ clicks are significant with $p\leq 0.001$, except for the results on MSLR-WEB30k with $\eta=1$ and $\epsilon^-_1=0.35$ with a significance of $p\leq 0.002$.
On small numbers of training clicks, the affine estimator has a similar or slightly lower performance than the other estimators.
This could be explained by the bias-variance tradeoff: the Bayes-IPS and IPS estimators could have lower variance due to their bias, making them perform better on small amounts of data.
Potentially, using propensity clipping on the affine estimator can increase its performance here~\citep{swaminathan2015batch}.

In conclusion, our results strongly indicate that optimizing with the affine estimator results in better performing ranking models than with previously proposed estimators.
In particular, on both datasets we see that, given enough click data, the affine estimator can be used to approximate the optimal ranking model, in settings with high or low degrees of trust bias or position bias.

{\renewcommand{\arraystretch}{0.01}
\begin{figure*}[t]
\centering
\begin{tabular}{l c c c c}
&
\small $\eta = 1$ and $\epsilon^-_1 = 0.65$
&
\small $\eta = 2$ and $\epsilon^-_1 = 0.65$
&
\small $\eta = 1$ and $\epsilon^-_1 = 0.35$
&
\small $\eta = 2$ and $\epsilon^-_1 = 0.35$
\\
\rotatebox[origin=lt]{90}{\hspace{2.4em} \small nDCG@10} &
\includegraphics[scale=0.3]{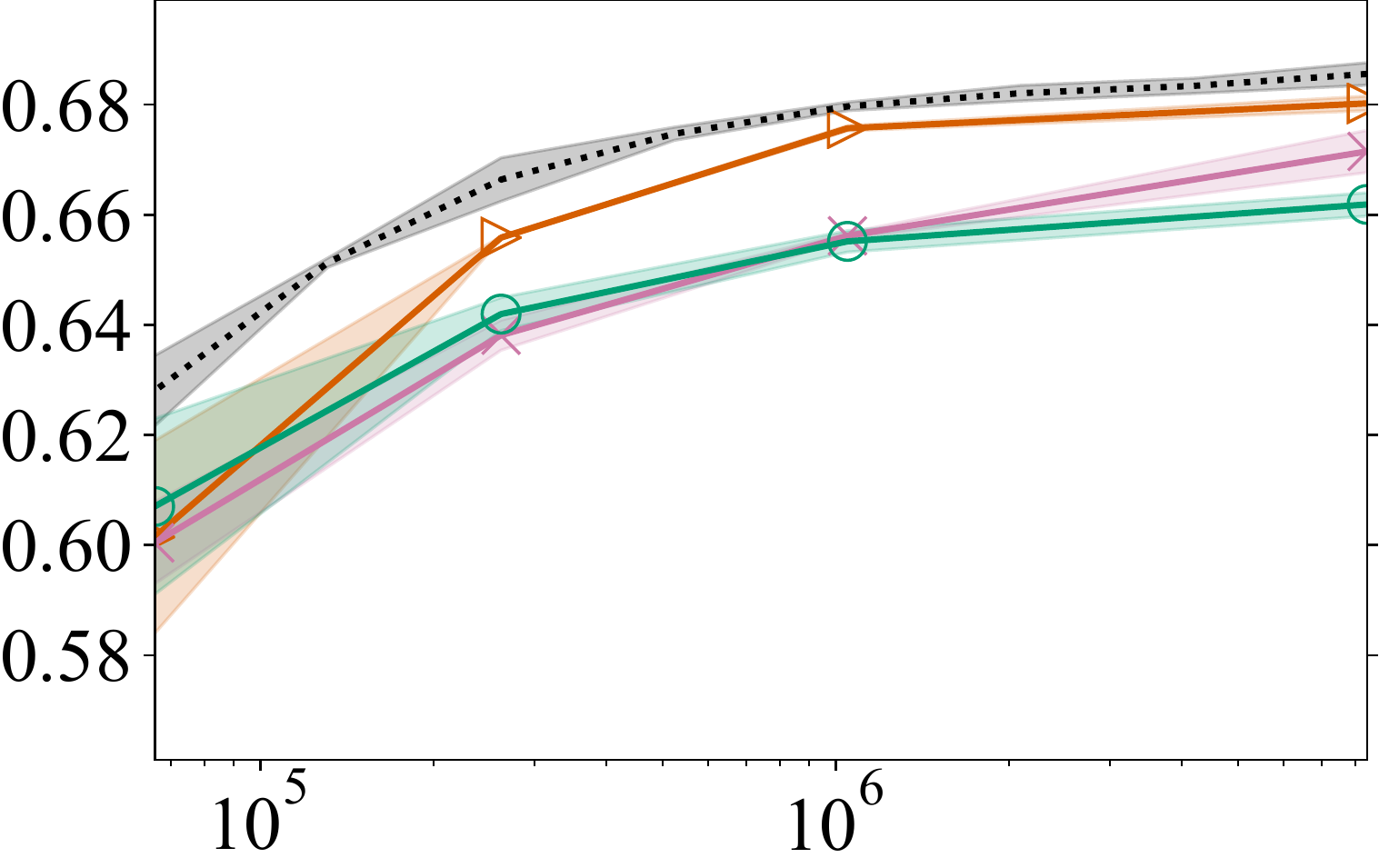} &
\includegraphics[scale=0.3]{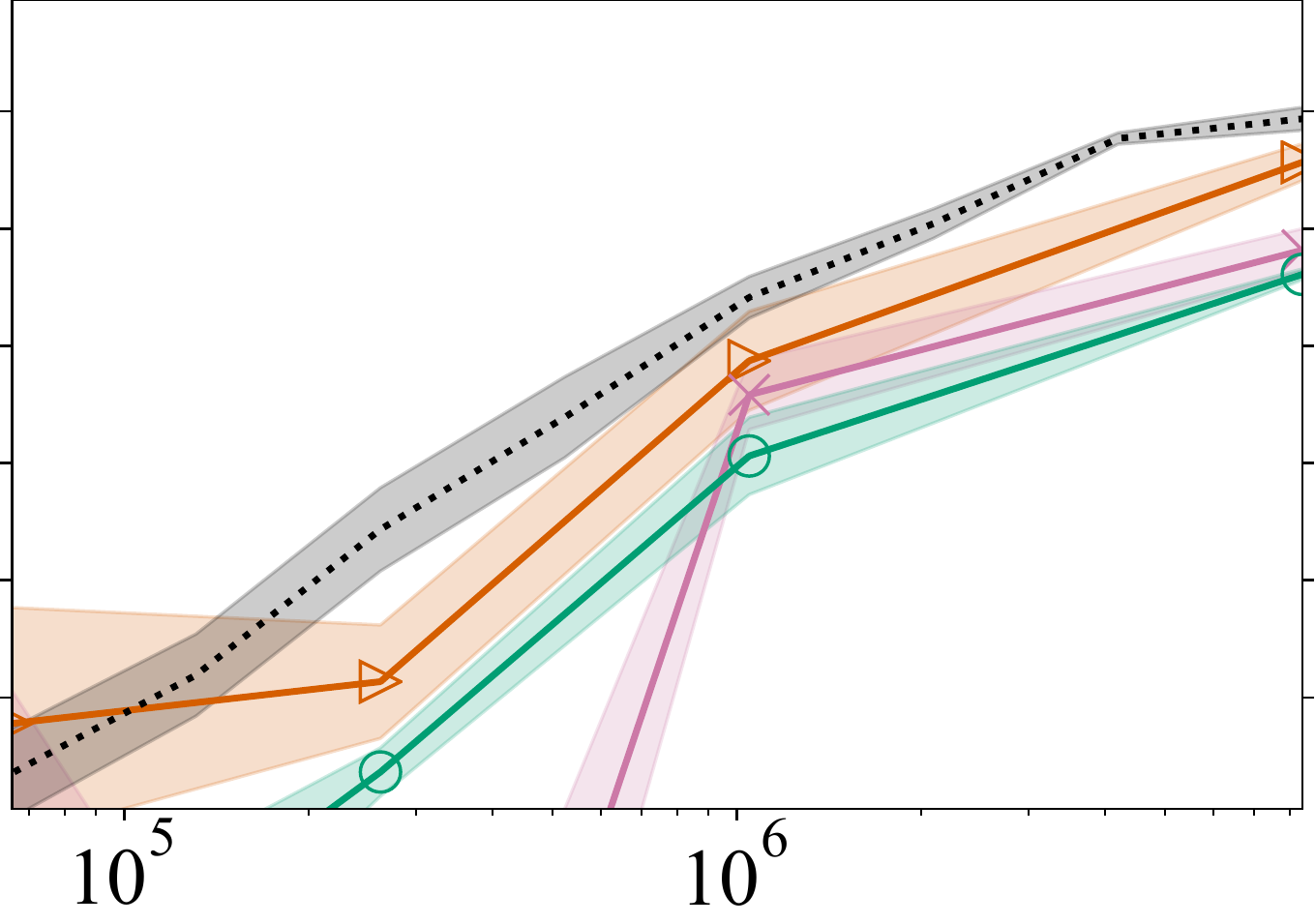} &
\includegraphics[scale=0.3]{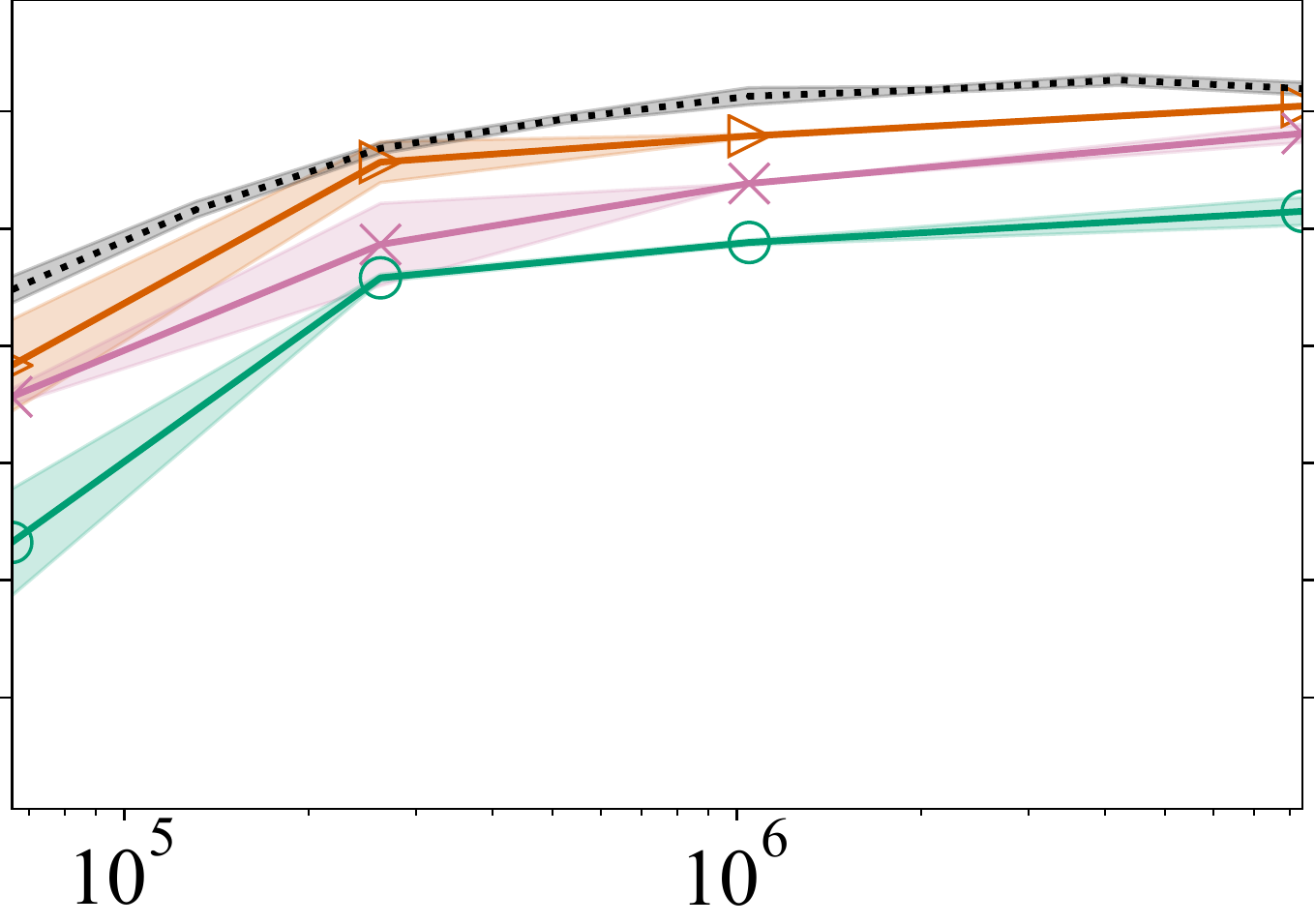} &
\includegraphics[scale=0.3]{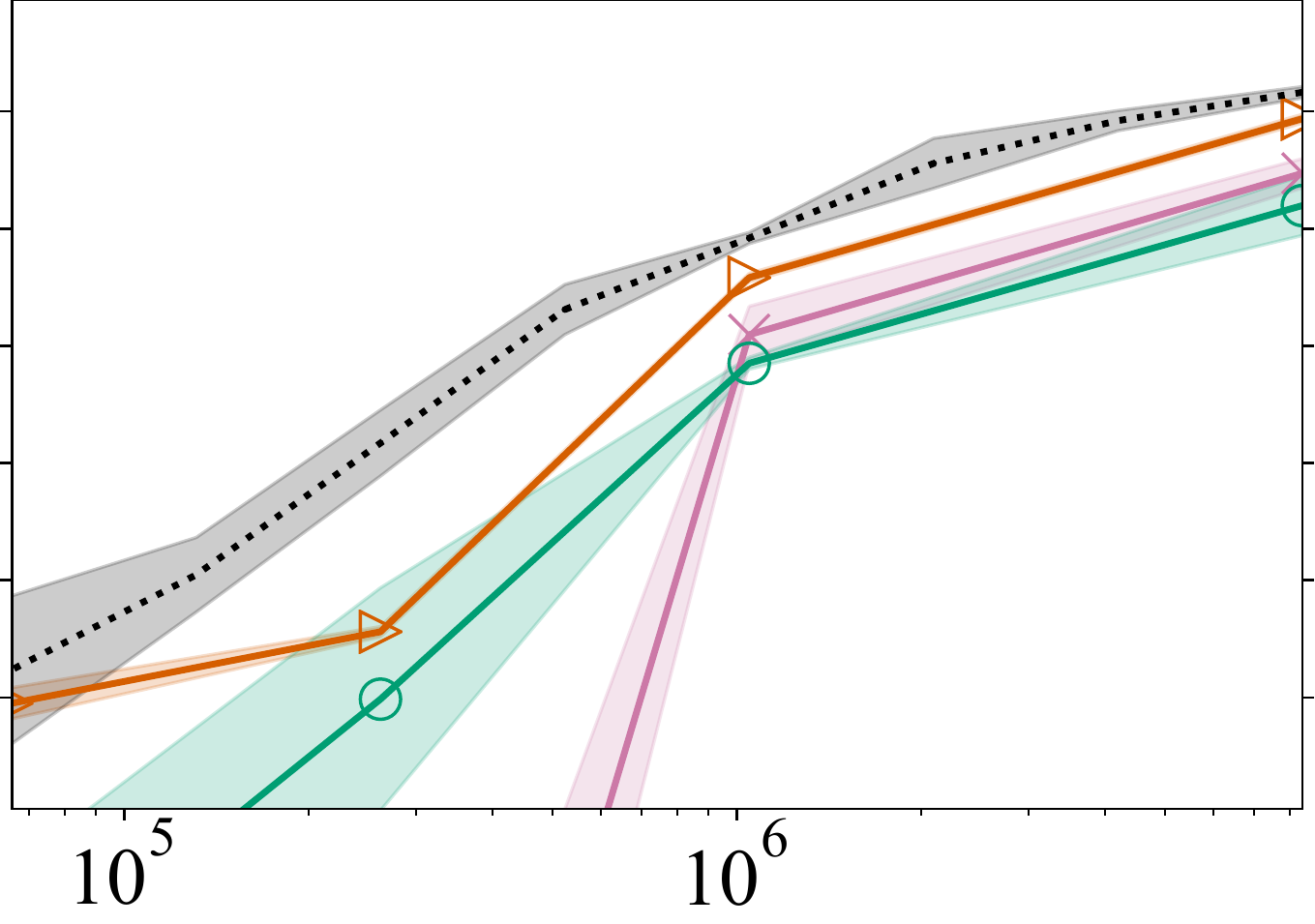}
\\
& \small Number of Training Clicks
& \small Number of Training Clicks
& \small Number of Training Clicks
& \small Number of Training Clicks
\\
 \multicolumn{5}{c}{
 \includegraphics[scale=.45]{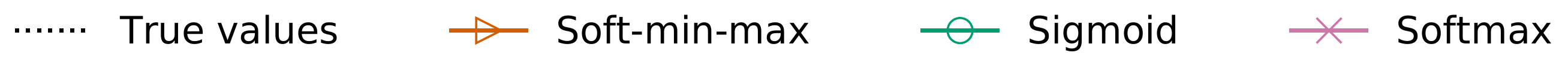}
} 
\end{tabular}
% \vspace{-0.5\baselineskip}
\caption{
Comparison of different final activation functions to estimate the bias parameters, under varying levels of position and trust bias.
Y-axis indicates the performance of ranking models optimized using the affine estimator. 
% For optimization, the estimator was given bias parameters estimated on the available click data.
Results are averaged over four runs; shaded area indicates the standard deviation.
All results are based on the Yahoo! Webscope dataset.
% \todo{Start all names with capital letters in the legend.}
% \harrie{I think its clearer if we just replace Oracle with True values here as well}
}
    \label{fig:EM}
\end{figure*}
}

\vspace*{-4mm}
\subsection{Optimization with estimated biases}
\label{subsec:EM}
Next, we consider whether \emph{optimization with the affine estimator is robust to estimated bias values.}
This is important as in practice the values of bias parameters have to be estimated as well.
While the theory proves that the affine estimator is unbiased when provided with the true bias values, we will now investigate whether it is still effective when they are estimated.

Figure~\ref{fig:EM} shows the performance (nDCG@10) reached by the affine estimator using bias parameters estimated from clicks (see Section~\ref{subsec:parametersestimation}), under varying degrees of position and trust bias.
For clarity, both the ranking model optimization and the bias parameter estimation used the same clicks.
Furthermore, the results in Figure~\ref{fig:EM} are separated for different final activation functions. 
Figure~\ref{fig:zeta} shows the estimated parameters after $8 \cdot 10^6$ clicks in the same settings.

In Figure~\ref{fig:EM} we see that parameter estimation with the soft-min-max function leads to the best performance: soft-min-max outperforms the other functions in all settings, regardless of the number of training clicks. 
Though the difference between soft-min-max and optimization with the true bias values is noticeable, it appears to be a small difference, especially after $10^6$ clicks.
This suggests that the affine estimator with the soft-min-max function is robust to estimated bias values.
Additionally, we see that the softmax function leads to decent performance when many clicks are available, but handles small numbers of clicks less well.
Lastly, the sigmoid function results in the poorest performance.
%, it appears that is the least preferable choice which leads to a considerable decrease in performance.

Interestingly, Figure~\ref{fig:zeta} shows that none of the functions leads to extremely accurate bias estimation with \ac{EM}.
We see that except for the first position, Soft-min-max and sigmoid underestimate the values of $\alpha_k$, while softmax overestimates it.
While soft-min-max and softmax have accurate estimates of $\beta_k$, sigmoid appears to underestimate it.
This further shows that the affine estimator is robust to estimated values, since soft-min-max leads to good performance while underestimating $\alpha_k$.
This seems to suggest that having an accurate estimate of $\beta_k$ is more important than one for $\alpha_k$.
In theory, from Eq.~\ref{eqn:affine_correction} we see that, unlike $\beta_k$, $\alpha_k$ can be estimated within a constant factor of the true value without hurting the performance of $\hat{\Delta}_\text{\OurMethodShort}$.
However, further analysis is required to fully understand what kind of inaccuracies still result in high performance.
These results also suggest that there are promising opportunities for novel ways to estimate trust bias from click data.

In conclusion, our results show that using the affine estimator still leads to good performance when it is based on estimated bias values.
In particular, we have found that using the soft-min-max function leads to the best results, and that the affine estimator can still get near-optimal performance when bias values are not completely accurate.
We conclude that the affine estimator is robust w.r.t.\ estimated bias values.

{\renewcommand{\arraystretch}{0.01}
\begin{figure*}[t]
\centering
\begin{tabular}{l c c c c}
&
\small $\eta = 1$ and $\epsilon^-_1 = 0.65$
&
\small $\eta = 2$ and $\epsilon^-_1 = 0.65$
&
\small $\eta = 1$ and $\epsilon^-_1 = 0.35$
&
\small $\eta = 2$ and $\epsilon^-_1 = 0.35$
\\ 
\rotatebox[origin=lt]{90}{\hspace{0.1em} \small (estimated) $\alpha_k$} &
\includegraphics[scale=0.3]{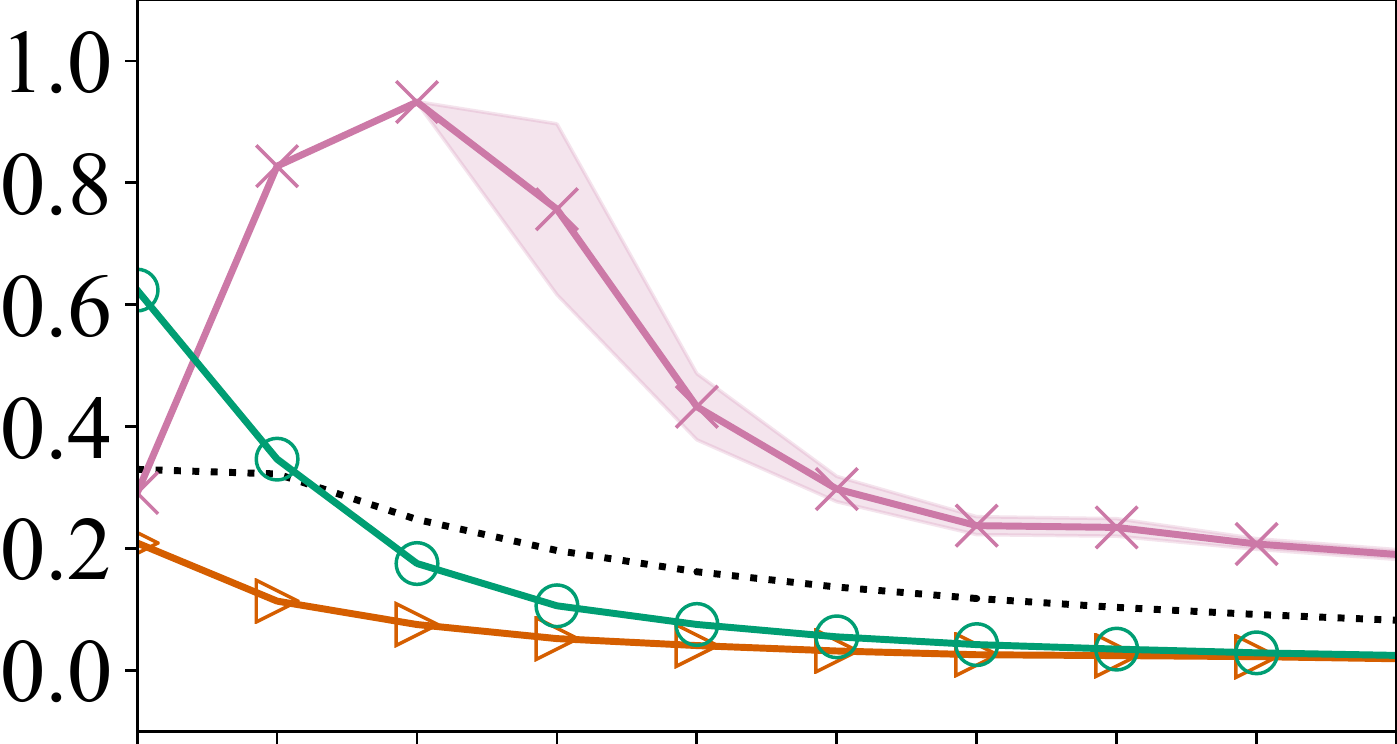} \hspace{0.08em} &
\includegraphics[scale=0.3]{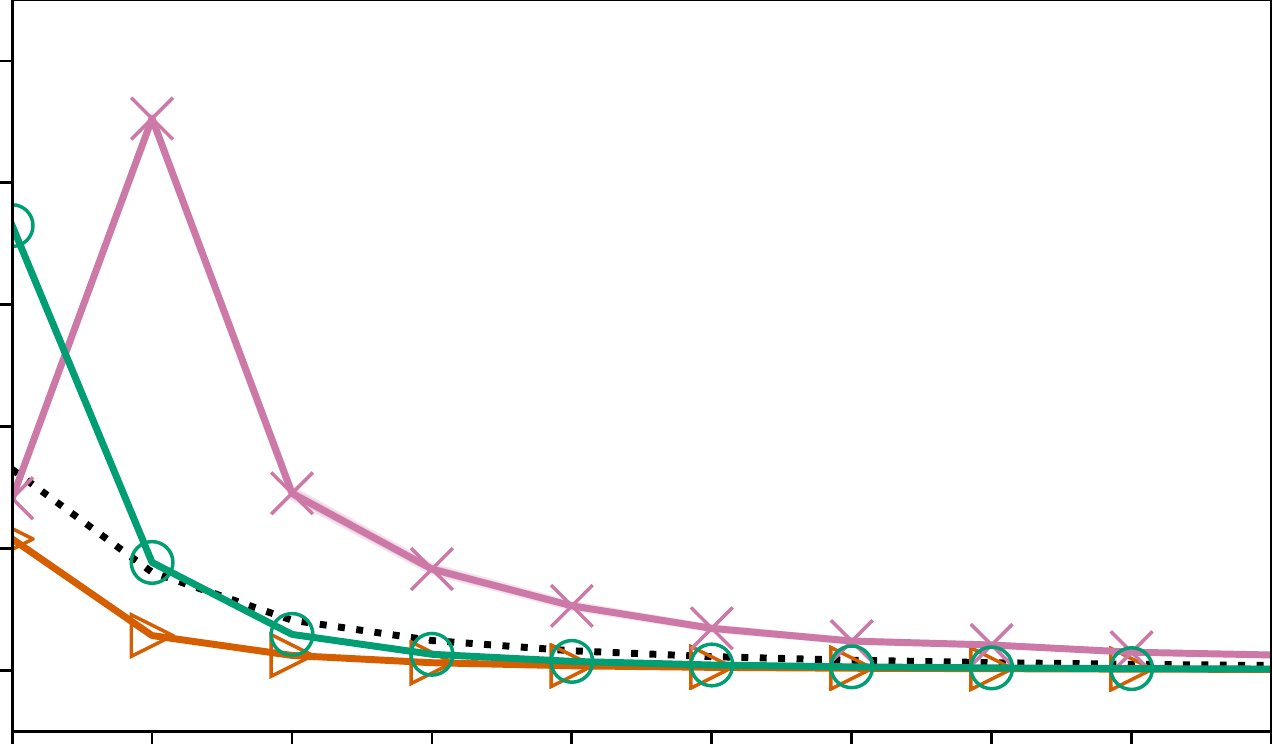} \hspace{0.0001em} &
\includegraphics[scale=0.3]{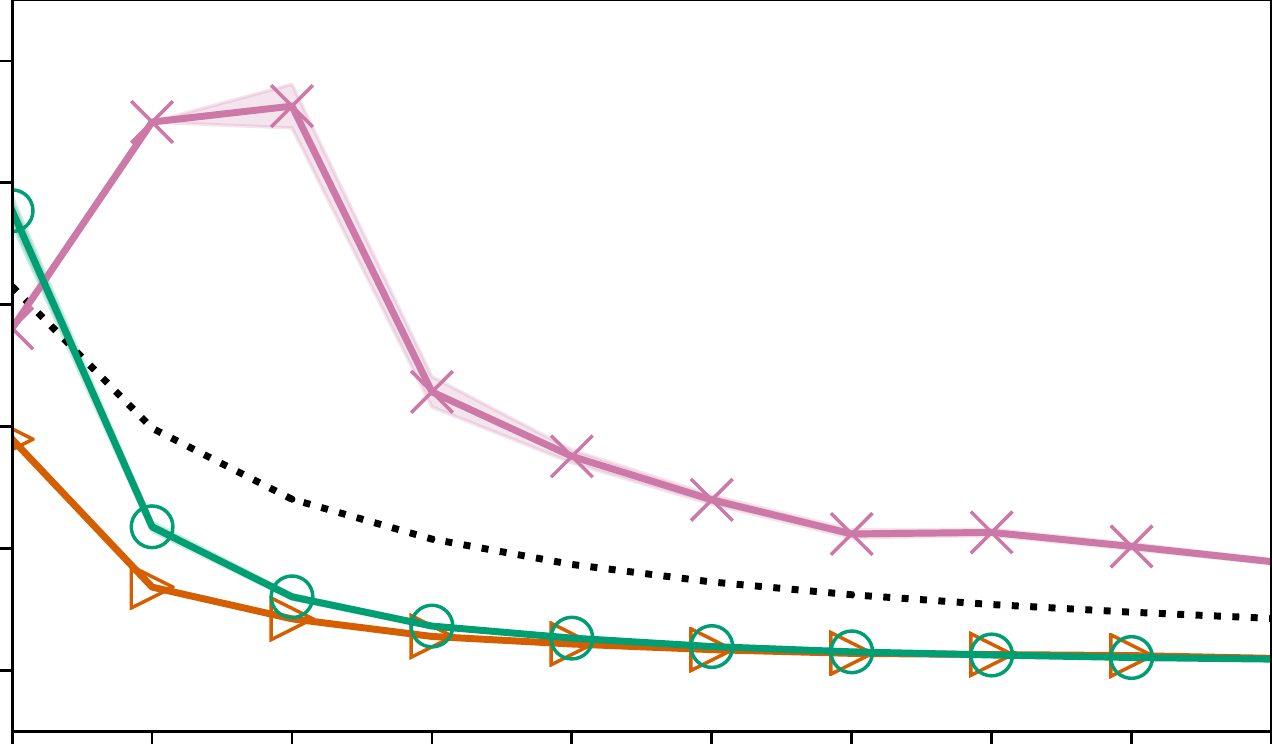} \hspace{0.0001em} &
\includegraphics[scale=0.3]{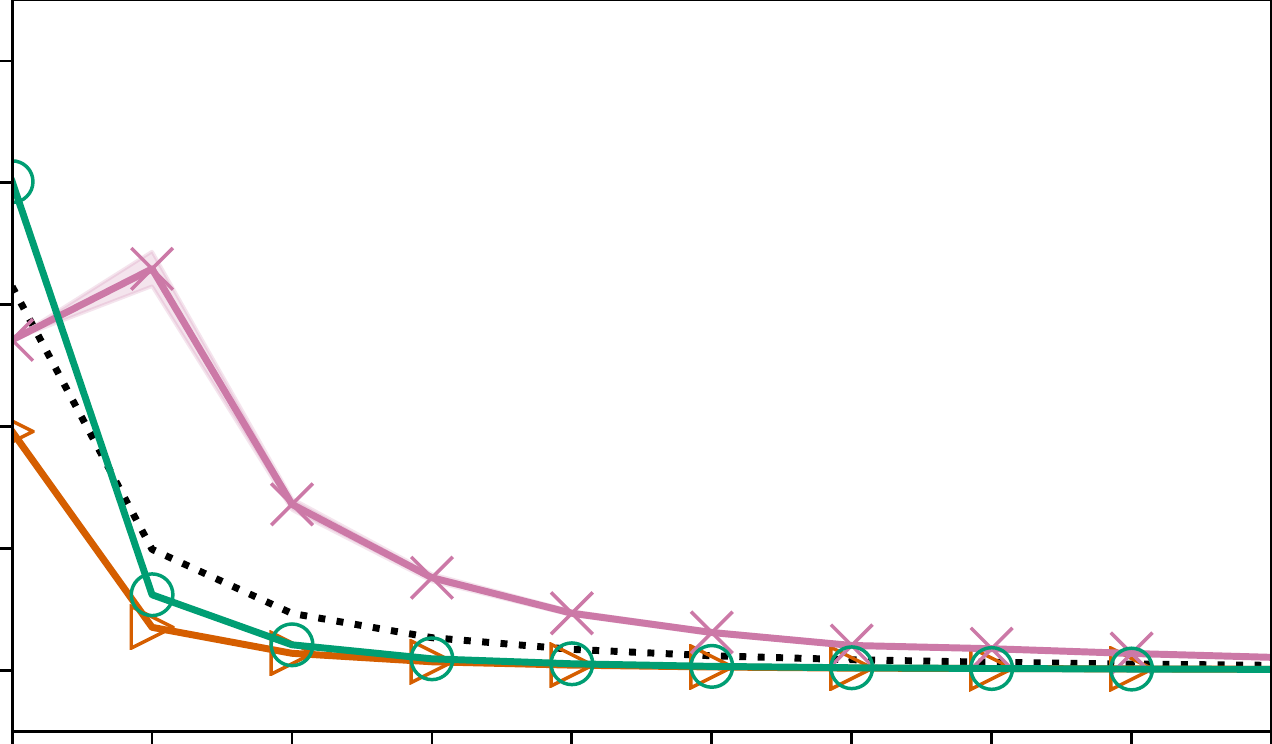} \hspace{0.0001em}
\\
\rotatebox[origin=lt]{90}{\hspace{1em} \small (estimated) $\beta_k$ }&
\includegraphics[scale=0.3]{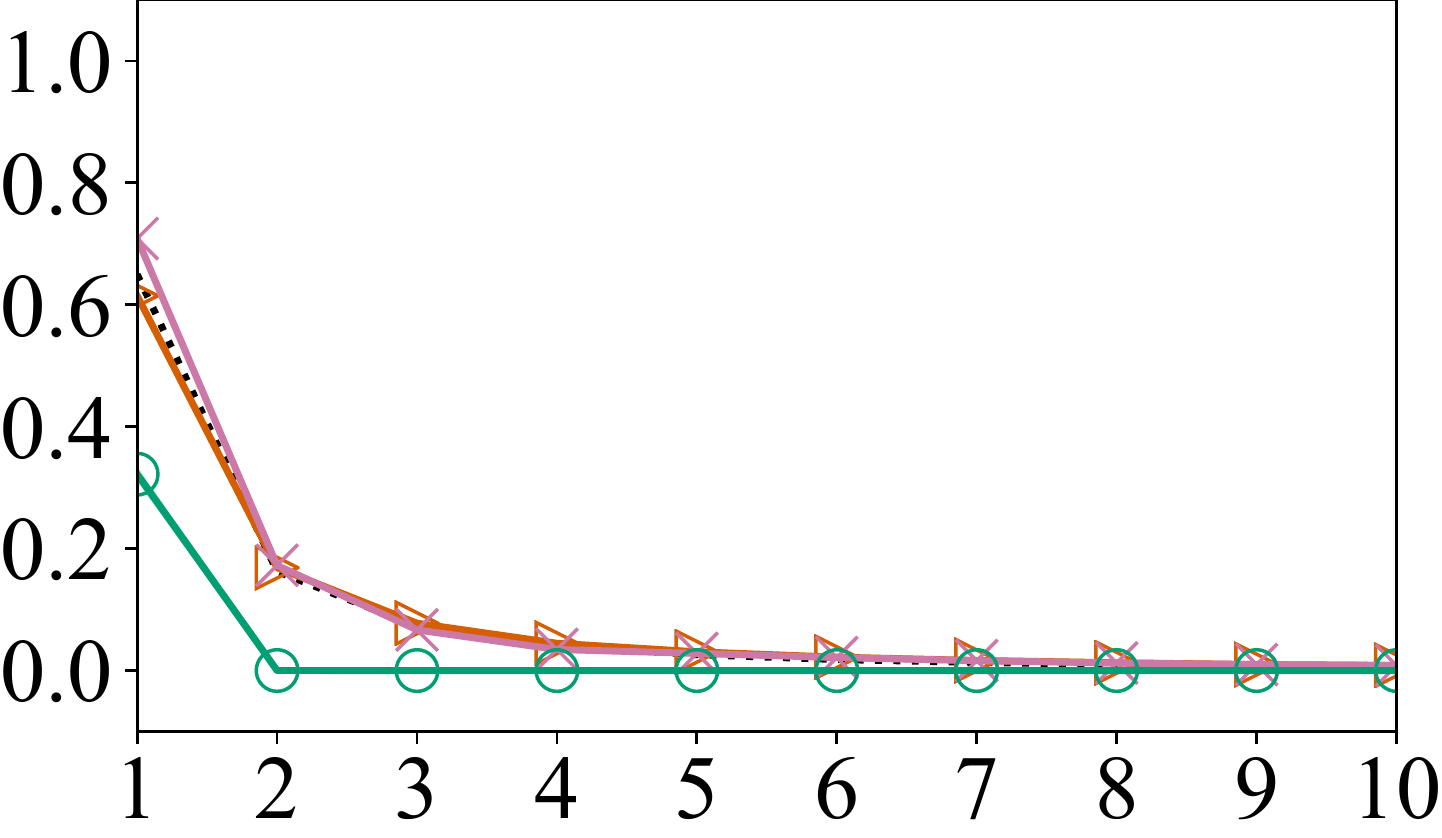} &
\includegraphics[scale=0.3]{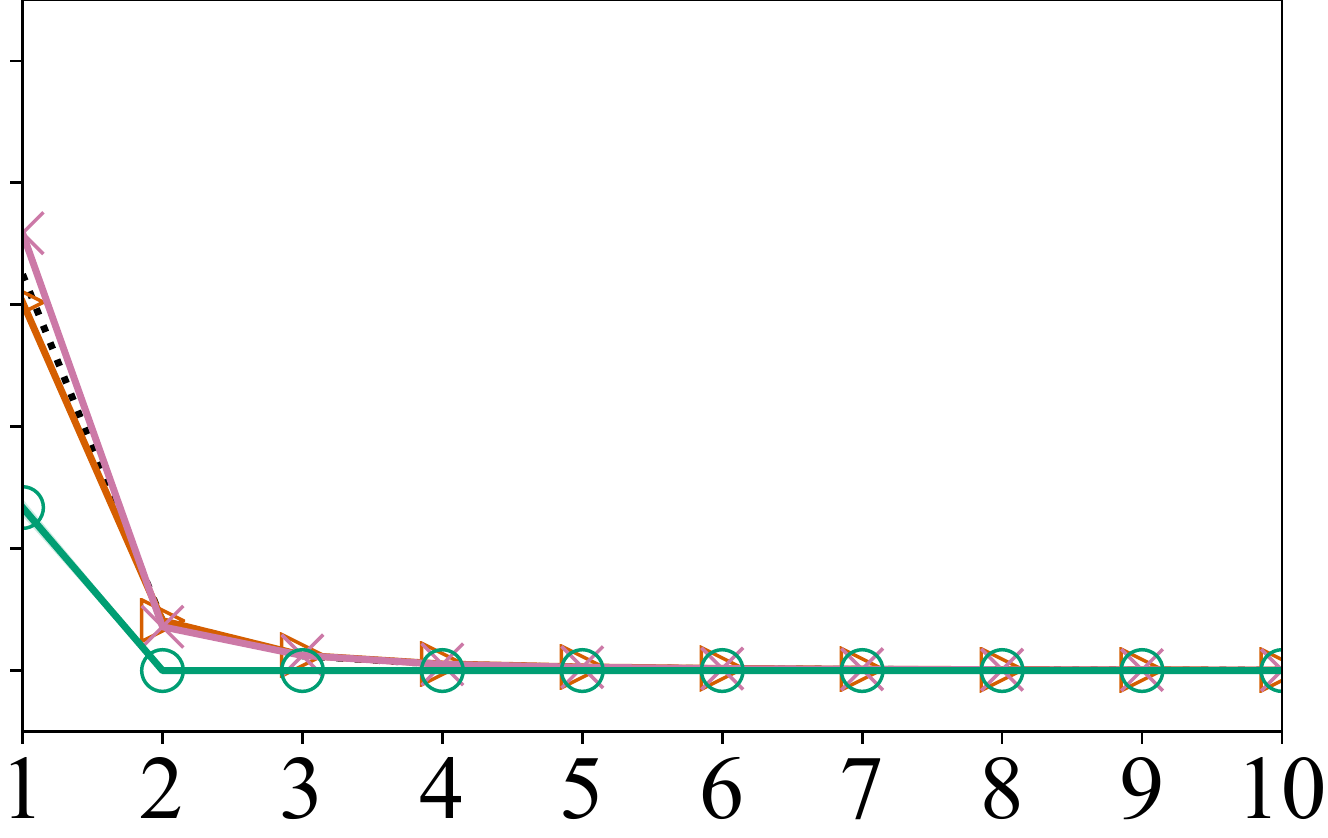} &
\includegraphics[scale=0.3]{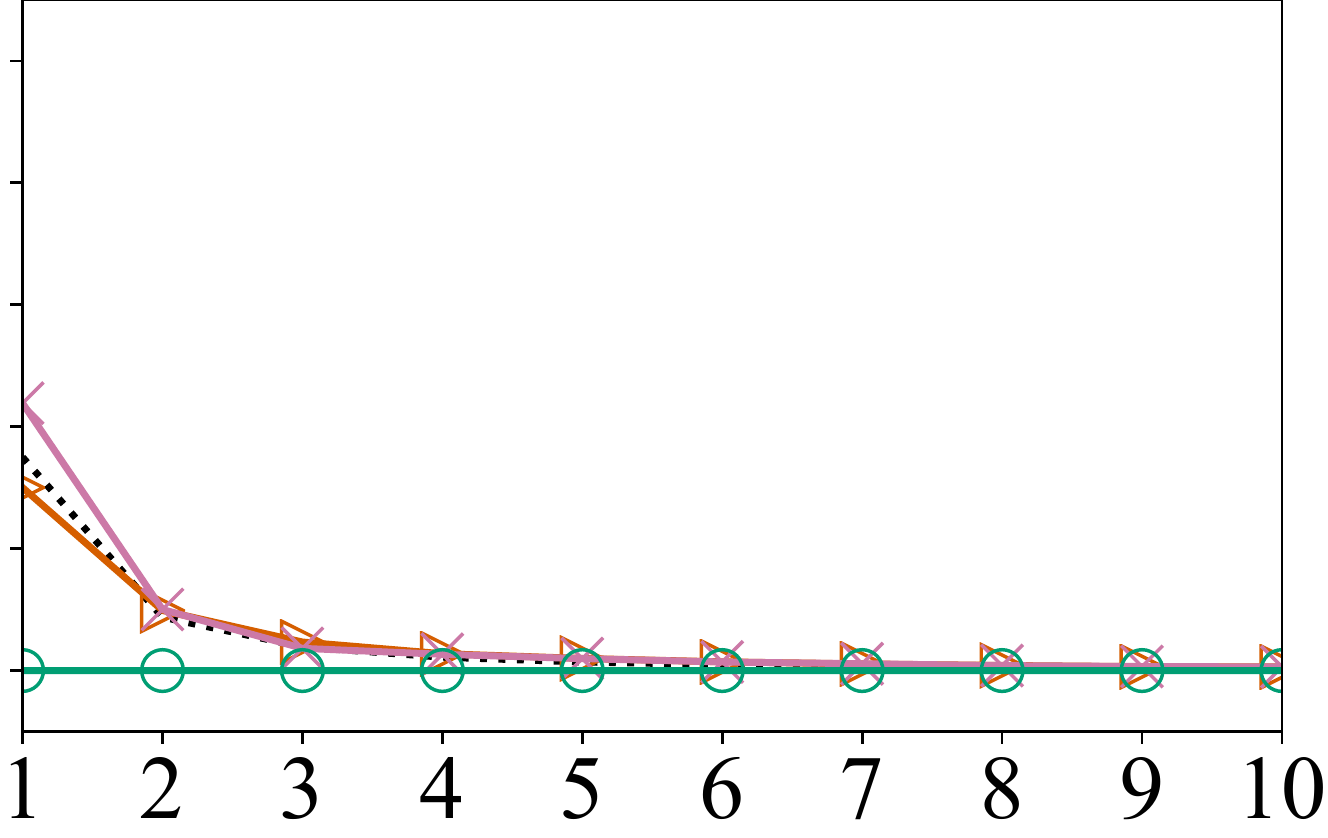} &
\includegraphics[scale=0.3]{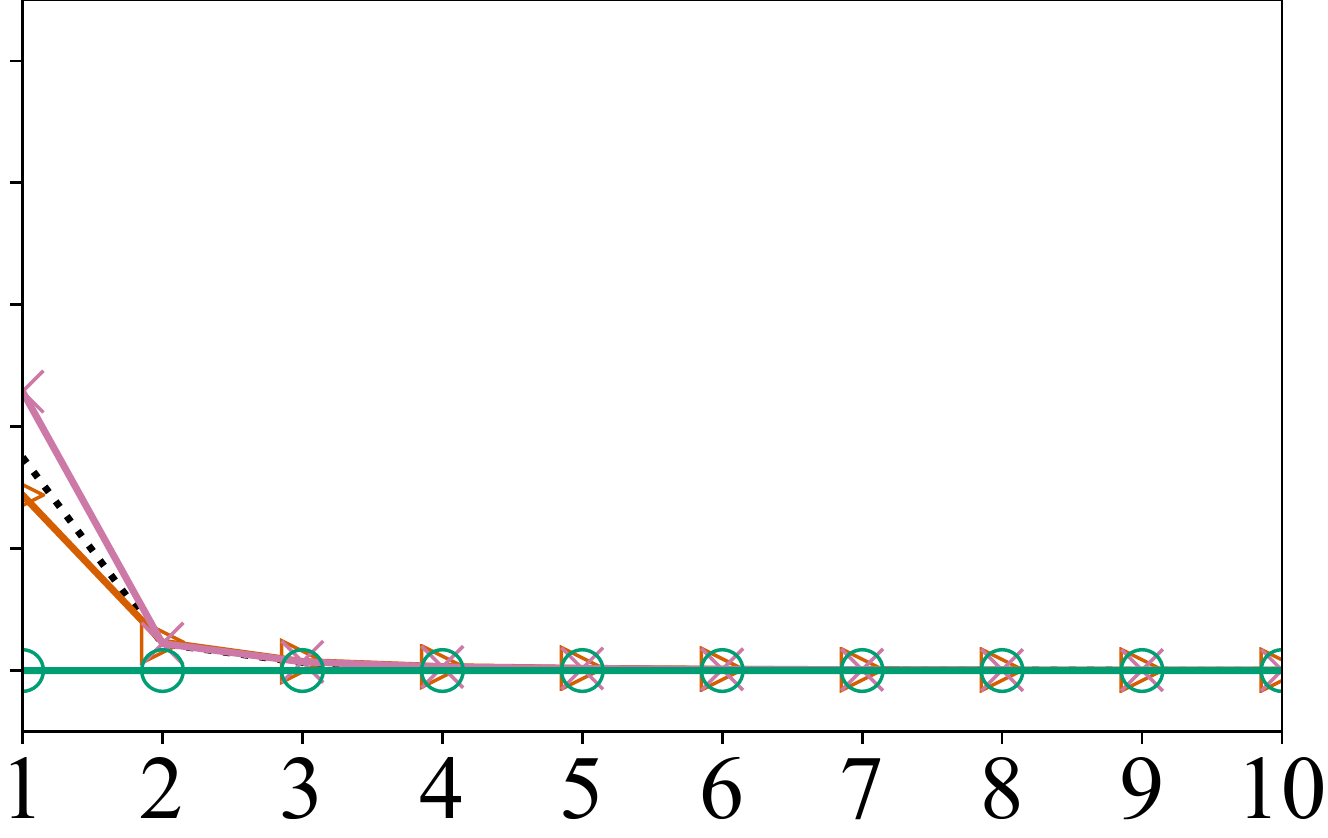}
\\
& \small Position ($k$)
& \small Position ($k$)
& \small Position ($k$)
& \small Position ($k$)
\\
 \multicolumn{5}{c}{
 \includegraphics[scale=.45]{sections/fig/zeta_legend}
} 
\end{tabular}
% \vspace{-0.5\baselineskip}
    \caption{
    Bias parameters estimated on $8\cdot 10^6$ clicks using different final activation functions, under varying degrees of position bias and trust bias.
Results are averaged over four runs; shaded area indicates the standard deviation. 
All results are based on the Yahoo! Webscope dataset. 
Top row: $\alpha_k$; bottom row: $\beta_k$.
% \todo{Start all names with capital letters in the legend.}
    }
    \label{fig:zeta}
\end{figure*}
}

%% file: sections/08-Conclusion.tex
% !TEX root = ../2020-trust-bias.tex
% \vspace*{-2.6mm}
\section{Conclusion}

In this paper we have considered \ac{CLTR} in situations with both position bias and trust bias.
We have proven that no \ac{IPS} estimator can correct for trust bias, including the Bayes-IPS estimator specifically designed for this bias~\citep{agarwal2019addressing}.
The reason for this inability is that trust bias is an affine transformation between relevance probabilities and click probabilities, and \ac{IPS} estimators can only correct for linear transformations.

As a solution, we have introduced the novel \OurMethod estimator, which applies affine transformations to clicks: it both reweights clicks and penalizes items for being displayed at ranks where the users' trust is high.
We proved that the \OurMethod estimator is unbiased w.r.t. both position bias and trust bias, thus it is the first \ac{CLTR} method that can deal with both of these biases simultaneously.
Furthermore, the \OurMethod estimator can be considered an extension of the existing \ac{IPS} approach: when no trust bias is present the \OurMethod estimator optimizes the same objective as the existing \ac{IPS} estimator.
Our experimental results show that using the \OurMethod estimator \ac{CLTR} can approximate the optimal model when both position bias and trust bias are present, while existing \ac{IPS}-based estimators cannot.
Furthermore, our results suggest that the estimator is robust to bias estimation, as performance is stable when the bias parameters are estimated from interactions.

With the introduction of our \OurMethod estimator, the \ac{CLTR} framework has been expanded to correct for trust bias on top of position bias.
Future work can continue this trend, for instance, by combining the policy-aware approach by \citet{oosterhuis2020topkrankings} with the \OurMethod estimator, perhaps an estimator that corrects for both item-selection bias and trust bias can be found.
Furthermore, previous work has found position bias estimation using randomization to be very powerful~\citep{agarwal2019estimating, wang2018position}.
Thus, there seems to be potential for methods based on randomization for estimating trust bias, possibly another fruitful direction for future research.